\documentclass{revtex4}
\usepackage{amsmath,amsthm}
\usepackage{rotating}
\usepackage{multirow}
\usepackage{graphicx}

\usepackage{amsfonts}
\newtheorem{theorem}{Theorem}[section]
\newtheorem{lemma}[theorem]{Lemma}
\newtheorem{proposition}[theorem]{Proposition}

\theoremstyle{remark}
\newtheorem{remark}[theorem]{Remark}

\theoremstyle{definition}
\newtheorem{definition}[theorem]{Definition}

\theoremstyle{example}

\theoremstyle{notation}

\newcommand{\bra}[1]{\langle#1|}
\newcommand{\ket}[1]{|#1\rangle}

\begin{document}

\title{Random projectors with continuous resolutions of the identity in a finite-dimensional Hilbert space}            
\author{A. Vourdas}
\affiliation{Department of Computer Science,\\
University of Bradford, \\
Bradford BD7 1DP, United Kingdom\\a.vourdas@bradford.ac.uk}

\begin{abstract}

Random sets are used to get a continuous partition of the cardinality of the union of many overlapping sets.
The formalism uses M\"obius transforms and adapts Shapley's methodology in cooperative game theory, into the context of set theory.
These ideas are subsequently generalized into the context of  finite-dimensional Hilbert spaces.
Using random projectors into the subspaces spanned by states from a total set, we construct an infinite number of continuous resolutions of the identity, that involve Hermitian positive semi-definite operators. 
The simplest one is the diagonal continuous resolution of the identity, and it is used to expand an arbitrary vector 
in terms of a continuum of components. It is also used to define the $F(x_1,x_2)$ function on the `probabilistic quadrant' $[0,\infty) \times [0,\infty)$,
which is analogous to the Wigner function for the harmonic oscillator, on the phase-space plane.
Systems with  finite-dimensional Hilbert space (which are naturally described with discrete variables) are described here with continuous probabilistic variables.
\end{abstract}
\maketitle

\section{Introduction}
\subsection{Background}

Coherent states and more generally positive operator valued measure (POVM) \cite{C1,C2,V2} play an important role in quantum mechanics.
Their most important property is the resolution of the identity which can be written as
\begin{eqnarray}\label{I1}
\sum _i \theta _i={\bf 1};\;\;\;i\in \Omega
\end{eqnarray}
where $\theta _i$ are operators, with indices in a set $\Omega$. 
The summation becomes integration in the case of a continuum of operators.
Using the resolution of the identity we can expand an arbitrary state in the Hilbert space $\ket{s}$ as
\begin{eqnarray}\label{I2}
\ket{s}=\sum _i \ket{s_i};\;\;\;\ket{s_i}=\theta _i\ket{s}.
\end{eqnarray}
In addition to that the quantity 
\begin{eqnarray}
f(i)={\rm Tr}(\rho \theta _i);\;\;\;\sum _{i \in \Omega}f(i)=1,
\end{eqnarray}
 where $\rho$ is a density matrix, is physically important. If $\theta _i$ are Hermitian operators, $f(i)$
is the average outcome of a measurement with $\theta _i$ on an ensemble described with the density matrix $\rho$.

In the case that $\theta _i$ are projectors, let $h_i$ be the subspace into which the $\theta _i$ projects.
Then we get a compartmentalisation  of the Hilbert space $H$ into subspaces $h_i$: 
\begin{eqnarray}
H=\bigvee _ih_i;\;\;\;h_i\wedge h_j={\bf 0}.
\end{eqnarray}
The conjunction $\wedge $ and disjunction $\vee$ of Hilbert spaces are defined briefly below (in Eqs(\ref{V1}), (\ref{V2})).
In general the subspaces $h_i$ are non-orthogonal to each other. Large values of $f(i)={\rm Tr}(\rho \theta _i)$
show the compartments  in the Hilbert space where most of the quantum state is located.
Known examples are:
\begin{itemize}
\item
A resolution of the identity in terms of orthogonal projectors.
For example, in the infinite-dimensional harmonic oscillator Hilbert space, projectors related to position states ($\Omega={\mathbb R}$) lead to a continuous resolution of the identity.
Also projectors related to number states ($\Omega={\mathbb Z}_0^+$) lead to a discrete resolution of the identity.
They lead to an expansion of a vector in an orthogonal basis, and 
the various components are independent from each other.
In these cases the $f(i)$ are probability distributions.
There is no redundancy in these cases, and this may be a disadvantage in noisy situations.

\item
A continuous resolution of the identity in terms of non-orthogonal projectors related to coherent states  ($\Omega={\mathbb C}$) in the infinite-dimensional harmonic oscillator Hilbert space.
This leads to an expansion of a vector in terms of components which are not independent from each other.
There is redundancy in this formalism, which can be an advantage in noisy situations.
The $f(i)={\rm Tr}(\rho \theta _i)$ is the $Q$ (or Husimi) function.
Large values of the $Q$-function, show the subspaces (related to coherent states) where most of the quantum state is located.

\end{itemize}

There are other resolutions of the identity where $\theta _i$ are {\bf not} projectors. 
In this case the Hilbert space is not compartmentalised into subspaces (the rank of the matrices $\theta _i$ might be equal to the dimension of the Hilbert space). 
However the expansion of Eq.(\ref{I1}) still holds, and the quantity $f(i)$ in Eq.(\ref{I2}) is physically important.
So it is not important the $\theta _i$ to be projectors.

An example is the
continuous resolution of the identity in terms of displaced parity operators in the infinite-dimensional harmonic oscillator Hilbert space.
In this case $\Omega={\mathbb C}$ (a very brief summary of this formalism is in section \ref{WWW} below).
The  $f(i)={\rm Tr}(\rho \theta _i)$ is then  the Wigner function.
Large values of the Wigner function, show which displaced parity operators overlap most with the quantum state.

The above examples, belong to the general area of `phase space methods'  which adapt the classical concept of phase space into a quantum context.
It is important to find new resolutions of the identity and generalise the above formalisms outside the context of phase space methods, but we emphasise that this can be a difficult task.
For this reason in the subject of frames and wavelets\cite{C3}, we have no exact resolution of the identity, but we have lower and upper bounds to it.

In \cite{Vou1,Vou2} we proposed another approach to this general area.
We started from an arbitrary total set of $n\ge d$ vectors in $d$-dimensional Hilbert space $H(d)$, which we call a pre-basis, and for which 
in general we have no resolution of the identity. We renormalised (`dressed') them into a `basis' of $n$ mixed states (density matrices), that resolve the identity
(in this case $\Omega=\{1,...,n\}$).
The renormalization formalism is inspired by the Shapley methodology in cooperative game theory\cite{G1,G2,G3,G4,G5} and uses the M\"obius transform\cite{R}.
But Shapley's approach is for scalar quantities, while our approach is for matrices.
The formalism leads to a {\bf discrete} resolution of the identity that involves a finite number of density matrices, and it is outside the general area of phase space methods.
We have shown that due to redundancy, the formalism is sensitive to physical changes and insensitive to noise.

\subsection{Present work}

In this paper we extend these ideas in a novel direction, by introducing random projectors into the $2^n$ subspaces spanned by vectors in a total set of $n$ vectors in $H(d)$. 
We note that there are many ways of definining random projectors and we use the definition \ref{def200}.
The average of these random projectors is a function of $n$ probabilities, and 
is used to get an infinite number of {\bf continuous} resolutions of the identity,  each of which involves a  one-dimensional  continuum
of Hermitian positive semi-definite operators. One of them, the diagonal continuous resolution of the identity, is simpler than the others and in this sense it is of special importance.
It is used to expand an arbitrary vector 
in terms of a continuum of components, and to define the $F(x_1,x_2)$ function on the `probabilistic quadrant' $[0,\infty) \times [0,\infty)$
which is analogous to the Wigner function for the harmonic oscillator.

Quantum systems with finite-dimensional Hilbert space use variables that take a finite number of values, and consequently they involve naturally many of the techniques in the general area of Discrete Mathematics: finite sums, finite fields, discrete Fourier transforms, etc (e.g. \cite{V1}). 
For example, their phase space consists of a finite number of points and it is described by finite geometries, the Wigner function takes a finite number of values, etc. 
The present approach describes these systems with continuous variables, that are related to the probabilities associated with random bases.
This enables the use of `Continuous Mathematics' in this context.
For example, we get a `continuous formalism' analogous to the phase space methods for the harmonic oscillator (which involves naturally continuous variables).

The Shapley methodology of cooperative game theory, is adapted here into the language of set theory. This makes it suitable for Physics applications, in contrast to cooperative game theory which is usually presented in the context of Mathematical Economics.
Furthermore exposition of these ideas in the simpler context of set theory, makes easier their generalization into Hilbert spaces which is our main objective.

For these reasons the paper is divided into two parts: the `set theory part' (sections 2, 4, 5) and the `Hilbert space part' (sections 3, 6, 7, 8).
The `set theory part' is interesting in its own right, and might have applications in other areas, but here we are interested in its generalization into 
finite-dimensional Hilbert spaces, which is presented in the `Hilbert space part'.
The overall presentation and notation in the paper, aims to make clear the analogy between the two parts.

The main results of the paper are: 
\begin{itemize}
\item
A continuous partition of the total cardinality of a finite number of overlapping sets (proposition \ref{PRO10}).
This uses probabilistic variables related to random sets (defined as in definition \ref{def100}).
An example is given in section \ref{ex}.
\item
Several continuous resolutions of the identity in a finite-dimensional Hilbert space  (proposition \ref{PRO100}), that involve Hermitian positive semi-definite operators.
This uses probabilistic variables related to random projectors (defined as in definition \ref{def200}).
Among them, the simplest one is the diagonal continuous resolution of the identity (proposition \ref{PRO200}). It is used to define the $F(x_1,x_2)$ function on the `probabilistic quadrant' $[0,\infty) \times [0,\infty)$,
which is analogous to the Wigner function for the harmonic oscillator, on the phase-space plane. An example is given in section \ref{ex1}.
\end{itemize}

We note that the general area of {\bf discrete random structures and  their relation to the continuum} is a `hot topic' in probability theory (e.g., random trees and graphs \cite{PR1,PR2}), in computer science (e.g.,  the journal `random structures and algorithms' ), in statistical physics and percolation \cite{PR3}, etc .
The present work brings related ideas into a quantum context.

\subsection{Contents}

In section 2 we present briefly the Shapley methodology of cooperative game theory, using the language of set theory. The main tool in this section is M\"obius transforms.
In section 3 we review briefly for later use, some aspects of the formalism in refs\cite{Vou1,Vou2}, which is based on Shapley's methodology, and which led to discrete resolutions of the identity in finite-dimensional Hilbert spaces.

 In section 4 we define random sets and study the properties of their probabilities. 
The average cardinality of these random sets is a function of $n$ probabilities, and 
is used in section 5 to get a continuous partition of the total cardinality of $n$ overlapping sets.

In section 6 we introduce random projectors  into the $2^n$ subspaces spanned by states in a total set of $n$ states.
The average of these random projectors is an operator (not a projector) that depends on a $n$ probabilities and is used to get an infinite number of continuous resolutions of the identity.
Among them the diagonal  continuous resolution of the identity plays a central role because of its simplicity. It is used in section 7 to expand an arbitrary vector, in terms of a continuum of components.
Also a function $F(x_1,x_2)$ is defined on the `probabilistic  quadrant' $[0, \infty)\times [0,\infty)$, which is analogous  to the Wigner function for the harmonic oscillator.
An example is presented in section 8. 

We conclude in section 9 with a discussion of our results.

\section{Shapley methodology in the context of set theory: dividing the overlaps between sets}\label{FFF}
In this section we adapt in the context of set theory, the Shapley methodology in cooperative game theory \cite{G1,G2,G3,G4,G5}.  
The players are replaced by overlapping sets, and the `worth' of each player (characteristic function) by the cardinality of each set.
A coalition is a union of some of these sets, and its `worth' is its cardinality.
The Shapley methodology divides the overlaps of the sets, equally  to all its `owners' (it is a type of `divorce settlement' for sets which have common `assets').
This leads to a new concept of cardinality, that we call Shapley cardinality (it shows the `assets' that belong exclusively to a set, after the `divorce settlement').

Our presentation of the Shapley methodology in the context of set theory is more appropriate for Physics,  
than cooperative game theory which is usually presented in the context of Mathematical Economics.

\subsection{M\"obius transform}
The M\"obius transform, is used extensively in Combinatorics, 
after the work by Rota\cite{R}. It is a generalization of
the inclusion-exclusion principle that gives the cardinality of the union of overlaping sets.
Rota generalized this to partially ordered structures.

Let $S_1,...S_n$ be a collection of finite sets, which in general overlap with each other. 
We call $\Omega$ the set of their indices $\Omega=\{1,...,n\}$.
If $A$ is a subset of $\Omega$, we denote as ${\mathfrak S}(A)$ 
the union of the sets with indices in $A$.
${\mathfrak S}(\Omega)$ is
the union of all the sets $S_1,...S_n$.
$\mu(A)$ the cardinality of ${\mathfrak S}(A)$:
\begin{eqnarray}\label{1}
{\mathfrak S}(A)=\bigcup _{i\in A} S_i;\;\;\;\mu(A)=|{\mathfrak S}(A)|;\;\;\;A\subseteq \Omega.
\end{eqnarray}
$|B|$ denotes the cardinality of a set $B$.
If $A=\emptyset$ then  ${\mathfrak S}(A)=\emptyset$ and $\mu (\emptyset)=0$.
Also ${\mathfrak S}(\{i\})=S_i$.
If $A\subseteq B$ then $\mu(A)\le \mu (B)$.
There are $2^n$ subsets of $\Omega$, and therefore there are $2^n$ values of $\mu (A)$.
In general
\begin{eqnarray}\label{ineq}
\mu (A)\ne \sum _{i\in A}\mu (\{i\}).
\end{eqnarray}

The M\"obius transform of $\mu (A)$ is defined as
\begin{eqnarray}\label{M}
{\mathfrak d} (A)=\sum  _{B\subseteq A}(-1)^{|A|-|B|} \mu(B).
\end{eqnarray}
For example,
\begin{eqnarray}\label{X1}
{\mathfrak d}(\{i\})&=&\mu (\{i\})\nonumber\\
{\mathfrak d} (\{i,j\})&=&\mu (\{i,j\})-\mu(\{i\})-\mu(\{j\})\nonumber\\
{\mathfrak d} (\{i,j,k\})&=&\mu (\{i,j,k\})-\mu(\{i,j\})-\mu(\{i,k\})-\mu(\{j,k\})+\mu(\{i\})+
\mu(\{j\})+\mu(\{k\}).
\end{eqnarray}

The inverse M\"obius transform is 
\begin{eqnarray}\label{b7}
\mu (A)=\sum _{B\subseteq A}{\mathfrak d} (B)=\sum _{i\in A}{\mathfrak d} (\{i\})+\sum _{i,j\in A}{\mathfrak d} (\{i,j\})+\sum _{i,j,k\in A}{\mathfrak d} (\{i,j,k\})+....
\end{eqnarray}
For example,
\begin{eqnarray}\label{X2}
\mu (\{i,j\})&=&{\mathfrak d} (\{i,j\})+{\mathfrak d}(\{i\})+{\mathfrak d}(\{j\})\nonumber\\
\mu (\{i,j,k\})&=&{\mathfrak d} (\{i,j,k\})+{\mathfrak d}(\{i,j\})+{\mathfrak d}(\{i,k\})+{\mathfrak d}(\{j,k\})+{\mathfrak d}(\{i\})+{\mathfrak d}(\{j\})+{\mathfrak d}(\{k\}).
\end{eqnarray}
Eqs(\ref{X1}),(\ref{X2}) show that the M\"obius transform describes the overlaps between sets, and the ${\mathfrak d} (\{i,j\})$, ${\mathfrak d} (\{i,j,k\})$,..., are `corrections' used to avoid double-counting.

For non-overlapping sets ($S_i\cap S_j=\emptyset$ for all $i,j$)
Eq.(\ref{ineq}) is equality for all $A$, and then ${\mathfrak d} (A)=0$ for all $A$ with cardinality greater or equal to $2$.
Therefore the importance of the ${\mathfrak d} (A)$ lies in the fact that Eq.(\ref{ineq}) is in general inequality.
\begin{remark}
Some of the work on game theory (e.g. \cite{G1}) is for superadditive game theory, where the whole is greater than the sum of its parts.
It is interesting that our use of game theory  in the context of set theory, is an example of subadditive game theory because
\begin{eqnarray}\label{3}
\mu (A)\le \sum _{i\in A}\mu (\{i\}).
\end{eqnarray}
In the language of game theory, this means that coalitions (i.e., set unions) play always a negative role (decrease the cardinality).
But our results are general and do not use Eq.(\ref{3}).

\end{remark}

\subsection{Partition of the total cardinality into Shapley cardinalities of the constituent sets}

We can rewrite Eq.(\ref{b7}) as follows:
\begin{eqnarray}\label{b89}
\mu (A)=\sum _{i\in A} M_A(i);\;\;\;M_A(i)=\mu (\{i\})+\frac{1}{2}\sum _{j\in A}{\mathfrak d} (\{i,j\})+\frac{1}{3}\sum _{j,k\in A}{\mathfrak d} (\{i,j,k\})+....
\end{eqnarray}
${\mathfrak d} (\{i,j\})$ is `common asset' that belongs to both sets $S_i, S_j$.
For this reason we add half of ${\mathfrak d} (\{i,j\})$ to the cardinality of the set $S_i$, and the other half  to the cardinality of the set $S_j$.
Similarly we add one third  of ${\mathfrak d} (\{i,j, k\})$ to the cardinality of the set $S_i$, another third to the cardinality of the set $S_j$, and another third to the cardinality of the set $S_k$, etc.

Eq.(\ref{b89}) is inspired by Shapley's methodology in cooperative game theory\cite{G2,G3,G4}, and is presented here in the context of set theory.
We call the $M_A(i)$ Shapley cardinality of the set $S_i$. 
It takes into account that the overlaps between sets are `joint property' which needs to be divided equally among all its owners. 
$M_A(i)$ depends not only on the set $S_i$ but also on its overlaps with the other sets with indices in $A$.

The cardinality of the union of sets, is the sum of the Shapley cardinalities of the various sets. 
For example, for two sets Eq.(\ref{X2}) can be rewritten as
\begin{eqnarray}\label{P1}
&&\mu (A)=M_A(i)+M_A(j);\;\;\;A=\{i,j\}\nonumber\\
&&M_A(i)=\mu(\{i\})+\frac{1}{2}{\mathfrak d} (A)\nonumber\\
&&M_A(j)=\mu(\{j\})+\frac{1}{2}{\mathfrak d} (A).
\end{eqnarray}
For three sets it can be written as
\begin{eqnarray}\label{P2}
&&\mu (A)=M_A(i)+M_A(j)+M_A(k);\;\;\;A=\{i,j,k\}\nonumber\\
&&M_A(i)=\mu(\{i\})+\frac{1}{2}[{\mathfrak d}(\{i,j\})+{\mathfrak d}(\{i,k\})]+\frac{1}{3}{\mathfrak d} (A)\nonumber\\
&&M_A(j)=\mu(\{j\})+\frac{1}{2}[{\mathfrak d}(\{i,j\})+{\mathfrak d}(\{j,k\})]+\frac{1}{3}{\mathfrak d} (A)\nonumber\\
&&M_A(k)=\mu(\{k\})+\frac{1}{2}[{\mathfrak d}(\{i,k\})+{\mathfrak d}(\{j,k\})]+\frac{1}{3}{\mathfrak d} (A).
\end{eqnarray}
We note that $M_A(i)$ depends on $A$, e.g., the $M_A(i)$ in Eq.(\ref{P1}) where $A=\{i,j\}$ is different from the the $M_A(i)$ in Eq.(\ref{P2}) where $A=\{i,j,k\}$.

In the special case of non-overlapping sets ($S_i\cap S_j=\emptyset$ for all $i,j$), the ${\mathfrak d} (A)=0$ for all $A$ with cardinality greater or equal to $2$. Then the $M_A(i)=\mu (\{i\})$ does not depend on $A$.

\subsection{An alternative approach to the Shapley cardinality}\label{sec20}

We order the sets as $S_{i_1},...,S_{i_n}$ and build gradually their union as
\begin{eqnarray}
S_{i_1} , {\mathfrak S}(\{i_1,i_2\}),...,{\mathfrak S}(\Omega).
\end{eqnarray} 
All $n!$ orders are equally likely. 
 Then $M_\Omega (i)$ is the average contribution to the cardinality, of the set $S_i$ as it joins the union of sets:
 \begin{eqnarray}\label{b90}
&&M_{\Omega} (i)=\sum _{A\subseteq \Omega} B(n-a+1, a)[\mu(A)-\mu(A-\setminus\{i\})]=\sum _{a=1}^nB(n-a+1, a)\lambda _{a}(i)\ge 0\nonumber\\
&&\lambda _{a}(i)=\sum _{|A|=a, A\ni i} [\mu(A)-\mu (A\setminus \{i\})]\ge 0;\;\;\;a=|A|\nonumber\\
&&\sum _{i=1}^nM_{\Omega} (i)=\mu(\Omega)
\end{eqnarray}
The equivalence between Eq.(\ref{b89}) (with $A$ replaced by $\Omega$) and Eq.(\ref{b90}) has been proved by Shapley\cite{G2,G3,G4,G5}.

The summation in the second equation is over all subsets of indices $A$, that contain $i$ and have cardinality $a$.
The $\mu(A)-\mu(A\setminus \{i\})$ is the `cardinality increment' as we go from the union of $(a-1)$ sets ${\mathfrak S}(A\setminus \{i\})$, to the union of sets ${\mathfrak S} (A)$.
The $\lambda _{a}(i)$  are an `average cardinality increment' when we add the set $S_i$ to the union ${\mathfrak S}(A\setminus \{i\})$.

$B(x,y)$ is the Beta function.
The Beta function in Eq.(\ref{b90}) is the probability that the set $S_i$ will join the above union after the sets with indices in $A\setminus \{i\}$ and before the 
sets with indices in $\Omega \setminus A$.
In order to see this, we first form the union 
\begin{eqnarray}
{\mathfrak S}(A\setminus \{i\})=\bigcup _{j\in A\setminus \{i\}}S_j.
\end{eqnarray}
There are $(a-1)!$ orders for doing this.
Then we form the union ${\mathfrak S}(A\setminus \{i\})\cup S_i={\mathfrak S}(A)$, and then the union
\begin{eqnarray}
{\mathfrak S}(A)\bigcup _{k\in \Omega\setminus A}S_k={\mathfrak S}(\Omega).
\end{eqnarray}
There are $(n-a)!$ orders for doing this last step. Since all $n!$ orders are equally likely, this leads to the probability
 \begin{eqnarray}\label{beta}
 \frac{(a-1)!(n-a)!}{n!}=B(n-a+1, a).
\end{eqnarray}
This probability is multiplied by $\lambda _{a}(i)$ which contains the contribution of the set $S_i$ to the cardinality $\mu(A-\setminus\{i\})$, for all $A$ that contain the element $i$ and have cardinality $|A|=a$.
\begin{remark}
The Shapley approach is one of many ways of dividing the overlaps. Other approaches like the Banzhaf method, have also been studied.
Later proposition \ref{PRO10} gives an infinite number of other ways.
\end{remark}

\subsection{A  generalized M\"obius-like relation between $\mu (A)$ and ${\mathfrak d}(A)$}\label{SEC1}

In this section we prove some relations which are needed later.
\begin{definition}
If $A,B$ are subsets of $\Omega$
\begin{eqnarray}
&&\delta (A,B)=1\;{\rm when}\;A=B\nonumber\\
&&\delta (A,B)=0\;{\rm when}\;A\ne B
\end{eqnarray}
\end{definition}
\begin{definition}
Let $q(A)$ where $A\subseteq \Omega$, be a set of $2^n$ real numbers.
Then
 \begin{eqnarray}\label{45A}
Q(A)=\sum _{B\supseteq A}q(B);\;\;\;{\mathfrak Q}_i=\sum _{A\ni i}q(A).
\end{eqnarray}
The following lemma considers two important special cases.
\begin{lemma}\label{L1}
\begin{itemize}
\item[(1)]
In the case $q(B)=\delta (B,C)$ we get
 \begin{eqnarray}
&&Q(A)=1\;{\rm when}\;A\subseteq C\nonumber\\
&&Q(A)=0\;{\rm otherwise}.
\end{eqnarray}
\item[(2)]
In the case  
\begin{eqnarray}
&&q(B)=(-1)^{|C|-|B|} \;{\rm when}\;B\subseteq C\nonumber\\
&&q(B)=0\;{\rm otherwise}.
\end{eqnarray}
we get $Q(A)=\delta(A,C)$.
\end{itemize}
\end{lemma}
\begin{proof}
\begin{itemize}
\item[(1)]
The proof of this is straightforward.
\item[(2)]
This reduces to 
\begin{eqnarray}
Q(A)=\sum_{A\subseteq B\subseteq C} (-1)^{|C|-|B|}.
\end{eqnarray}
If $C=A$ then $Q(A)=1$.
If $C$ is any other set, e.g., $C=A\cup \{i\}$ or $C=A\cup\{i,j\}$, etc, it is easily seen that $Q(A)=0$. 
\end{itemize}
\end{proof}
\end{definition}
\begin{lemma}\label{PRO1}
 \begin{eqnarray}\label{X}
\sum _A\mu(A)q(A)=\sum _A{\mathfrak d}(A)Q(A).
\end{eqnarray}
In the special case of non-overlapping sets ($S_i\cap S_j=\emptyset$ for all $i,j$)
 \begin{eqnarray}\label{X}
\sum _A\mu(A)q(A)=\sum _A{\mathfrak d}(A)Q(A)=\sum _i \mu (\{i\}) {\mathfrak Q}_i.
\end{eqnarray}
\end{lemma}
\begin{proof}
From Eq.(\ref{b7}) we get
\begin{eqnarray}
\sum _{A}\mu (A)q(A)=\sum _{A}\sum _{B\subseteq A}{\mathfrak d} (B)q(A)=\sum _{B}{\mathfrak d} (B)\left [\sum _{A\supseteq B}q(A)\right]=
\sum _{B}{\mathfrak d} (B)Q(B).
\end{eqnarray}
Eq.(\ref{45A}) has been used in the last step.

In the special case of non-overlapping sets, Eq.(\ref{ineq}) becomes equality and 
\begin{eqnarray}
\sum _{A}\mu (A)q(A)=\sum _i\sum _{A} \mu (\{i\}) q(A)=\sum _i \mu(\{i\})\left [\sum _{A\ni i}q(A)\right ]=\sum _i \mu(\{i\}){\mathfrak Q}_i.
\end{eqnarray}
\end{proof}
Using lemma \ref{L1} we show that the M\"obius transform in Eqs(\ref{M}), and the inverse M\"obius transform in
Eq.(\ref{b7}) are special cases of the more general relation in Eq.(\ref{X}).

\section{Discrete resolutions of the identity in a finite-dimensional Hilbert space}

\subsection{M\"obius transform for projectors}\label{A}
\begin{definition}
A set of vectors in a Hilbert space is total, if there is no vector which is orthogonal to all vectors in the set.
\end{definition}
\begin{definition}
A `pre-basis' in a $d$-dimensional Hilbert space $H(d)$  is a set of $n\ge d$ states 
\begin{eqnarray}\label{sig}
\Sigma =\{\ket{i}\;|\;i\in \Omega\};\;\;\;\Omega=\{1,...,n\}
\end{eqnarray}
such that:
\begin{itemize}
\item
Any subset of $d$ of these states, are linearly independent. 
\item
$\Sigma$ and also any of its subsets with $r\ge d$ of these states, are total sets.
\item
In general, we have no resolution of the identity in terms of these $n$ states.
\end{itemize}
\end{definition}
We emphasize from the outset that there is redundancy in the sense that the pre-basis has more vectors than the dimension of the space.
Redundancy is essential in noisy situations, and this is indeed the merit of this approach.

Let $h_1, h_2$ be two subspaces of $H(d)$. Their disjunction is
\begin{eqnarray}\label{V1}
h_1\vee h_2={\rm span}(h_1\cup h_2).
\end{eqnarray}
This is the quantum OR operation and includes all superpositions of vectors in the two spaces (unlike the Boolean OR which is simply the union of sets).
Their conjunction is the logical AND
\begin{eqnarray}\label{V2}
h_1\wedge h_2=h_1\cap h_2.
\end{eqnarray}

Let $H(\{i\})$ be the one-dimensional subspace that contains the vector $\ket{i}$, 
and $H(A)$ be the subspace spanned by all the states $\ket{i}$ with $i\in A\subseteq \Omega$:
\begin{eqnarray}
H(A)=\bigvee _{i\in A}H(\{i\}).
\end{eqnarray}
We call $\Pi[H(A)]$ or for simplicity $\Pi(A)$ the projector to the subspace $H(A)$. In particular
\begin{eqnarray}
\Pi(\{i\})=\ket{i}\bra{i};\;\;\;\Pi(\emptyset)=0.
\end{eqnarray}
There are $2^n$ projectors $\Pi(A)$. 
If $|A|\ge d$ then $\Pi(A)={\bf 1}$, so some of these $2^n$ projectors are equal to ${\bf 1}$.
Also
\begin{eqnarray}
&&{\rm Tr}[\Pi(A)]=|A|\;\;{\rm if}\;\;|A|<d\nonumber\\
&&{\rm Tr}[\Pi(A)]=d\;\;{\rm if}\;\;|A|\ge d.
\end{eqnarray}

In general
\begin{eqnarray}\label{377}
\Pi(A)\ne \sum _{i\in A}\Pi(\{i\}).
\end{eqnarray}
Only if the kets $\ket{i}$ where $i\in A$ are orthogonal to each other, we get equality in this equation.

In practical calculations, the projectors $\Pi(A)$ can be calculated as follows. 
We express the vectors $\ket{i}$ where $i\in A$, as $d\times 1$ columns, and then write the $d\times |A|$ matrix ${\mathfrak A}$ which has as columns these vectors.
The projector $\Pi(A)$ is given by
\begin{eqnarray}\label{670A}
\Pi(A)={\mathfrak A}({\mathfrak A}^\dagger {\mathfrak A})^{-1}{\mathfrak A}^\dagger.
\end{eqnarray}

The M\"obius transform of the projectors $\Pi(A)$, is given by:
\begin{eqnarray}\label{m11}
{\mathfrak D} (B)=\sum _{A\subseteq B} (-1)^{|A|-|B|}\Pi(A);\;\;\;\;A,B\subseteq\Omega.
\end{eqnarray}
The inverse M\"obius transform is
\begin{eqnarray}\label{m13}
\Pi (A)=\sum _{B\subseteq A}{\mathfrak D} (B)=\sum _{i\in A}\Pi(\{i\})+\sum _{i,j\in A}{\mathfrak D} (\{i,j\})+\sum _{i,j,k\in A}{\mathfrak D} (\{i,j,k\})+....
\end{eqnarray}

If Eq.(\ref{377}) is equality (which occurs in the case of an orthonormal basis), then ${\mathfrak D} (A)=0$ for all subsets with cardinality greater or equal to $2$.
Therefore the importance of the ${\mathfrak D} (A)$ lies in the fact that Eq.(\ref{377}) is inequality.

The analogue of the generalized M\"obius like relation in Eq.(\ref{X}) is in the present context:
 \begin{eqnarray}\label{XX}
\sum _A\Pi(A)q(A)=\sum _A{\mathfrak D}(A)Q(A).
\end{eqnarray}

The projectors $\Pi(A)$ in Eq.(\ref{m11}) do not commute (in general) with each other.
Given a density matrix $\rho$, the various ${\rm Tr}[\rho \Pi(A)]$ can be measured using different ensembles (described by the same density matrix $\rho$), for the various projectors.
Then the expectation values ${\rm Tr}[\rho{\mathfrak D} (B)]$ of the
 M\"obius operators ${\mathfrak D} (B)$ can be calculated.

\subsection{Discrete resolutions of the identity in terms of density matrices}

In refs\cite{Vou1,Vou2}, inspired by Shapley's methodology in cooperative game theory, we have rewritten Eq.(\ref{m13}) as
\begin{eqnarray}\label{qaz}
\Pi (A)=\sum _{i\in A}\theta _A(i);\;\;\;\theta _A(i)=\Pi(\{i\})+\frac{1}{2}\sum _{j\in A}{\mathfrak D} (\{i,j\})+\frac{1}{3}\sum _{j,k\in A}{\mathfrak D} (\{i,j,k\})+....
\end{eqnarray} 
As in Eq.(\ref{b89}), the joint parts of the subspaces $H(\{i\})$ quantified with the operators ${\mathfrak D} (\{i,j\}), {\mathfrak D} (\{i,j,k\})$, etc, are divided 
equally among all its `owners'.
Using terminology from quantum field theory, we can say that Eq.(\ref{qaz}) starts with the `bare' operators $\Pi(\{i\})$ and produces the `dressed' or renormalized operators $\theta _A(i)$.
If $A=\Omega$ we get
$\sum _{i\in \Omega}\theta _{\Omega}(i)={\bf 1}$.
We have shown in ref.\cite{Vou1,Vou2} that the 
\begin{eqnarray}\label{ppp}
&&\sigma _{\Omega}(i)=\frac{n}{d}\theta _{\Omega}(i);\;\;\;\frac{d}{n}\sum _{i=1}^n\sigma _{\Omega}(i)={\bf 1};\;\;\;i=1,...,n
\end{eqnarray}
are $n$ density matrices, which do not commute with each other, and which resolve the identity.
We have also shown that they can be written as
\begin{eqnarray}\label{rrr}
&&\sigma _{\Omega}(i)=\frac{n}{d}\sum _{a=1}^d\Lambda _{a}(i)B(a, n-a+1)\nonumber\\
&& \Lambda _{a}(i)=\sum _{|A|=a, A\ni i} [\Pi(A)-\Pi (A\setminus \{i\})];\;\;\;{\rm Tr} \Lambda _{a}(i)=
\begin{pmatrix}
n-1\\
a-1
\end{pmatrix}
\end{eqnarray}
The $\Pi(A)-\Pi(A\setminus \{i\})$ are projectors which give the `increment' as we go from the $(a-1)$-dimensional subspace $H(A\setminus \{i\})$  to the  $a$-dimensional subspace $H(A)$ which contains the vector $\ket{i}$.
The $\Lambda _{a}(i)$  are an `average increment' when we add the vector $\ket{i}$ to an $(a-1)$-dimensional subspace with vectors in the pre-basis other than $\ket{i}$.

The $ \Lambda _{a}(i)$ are Hermitian positive semi-definite matrices which are not projectors, and which for a given pre-basis are fixed. 
The summation over $a$ is up to $d$, because for $a\ge d+1$ we get $ \Lambda _{a}(i)=0$.

Eqs.(\ref{ppp}),(\ref{rrr}) are analogous to Eq.(\ref{b90}) and we interpret it as follows.
We order the subspaces $H(\{i_1\}),...,H(\{i_n\})$ and build gradually their disjunctions as 
\begin{eqnarray}
H(\{i_1\}), H(\{i_1, i_2\}), ...,H(d).
\end{eqnarray} 
All $n!$ orders are equally likely. 
The corresponding projectors are
\begin{eqnarray}
\Pi(\{i_1\}), \Pi(\{i_1, i_2\}), ...,{\bf 1}.
\end{eqnarray} 
Then $\frac{d}{n}\sigma_\Omega (i)$ is the average contribution of $\Pi(\{i\})$ to ${\bf 1}$.
In particular the Beta function in Eq.(\ref{ppp}) is the probability that $H(\{i\})$ will join the above disjunction after the subspaces with indices in $A\setminus \{i\}$ and before the 
subspaces with indices in $\Omega \setminus A$.

We  see this with an argument analogous to the one in section \ref{sec20} for sets. We first form the disjunction 
\begin{eqnarray}
H(A\setminus \{i\})=\bigvee _{j\in A\setminus \{i\}}H(\{j\}).
\end{eqnarray}
There are $(a-1)!$ orders for doing this.
Then we form the disjunction $H(A\setminus \{i\})\vee H(\{i\})=H(A)$, and then the disjunction 
\begin{eqnarray}
H(A)\bigvee _{k\in \Omega\setminus A}H(\{k\})=H(d).
\end{eqnarray}
There are $(n-a)!$ orders for doing this last step. Since all $n!$ orders are equally likely, this leads to the probability in Eq.(\ref{beta}) that involves the Beta function.
This probability is multiplied by the matrix $\Lambda _{a}(i)$ which contains the contribution of $\Pi(i)$ to $\Pi(A-\setminus\{i\})$, for all $A$ with $|A|=a$ that contain the element $i$.

Using the resolution of the identity in Eq.(\ref{ppp}), we can expand an arbitrary vector $\ket{s}$ in this basis:
\begin{eqnarray}\label{vb}
\ket{s}=\sum _{i=1}^n\ket{s_i};\;\;\;\ket{s_i}=\frac{d}{n} \sigma _{\Omega}(i)\ket{s}=\sum _{a=1}^n\left [B(a,n-a+1)\Lambda _{a}(i)\ket{s}\right ].
\end{eqnarray}
We have shown  in ref.\cite{Vou1,Vou2} that such an expansion is sensitive to physical changes and insensitive to noise.
This is due to the redundancy in the formalism.

 \section{Random sets and their probabilities}\label{GGG}

We consider the set $\Omega$ that contains the indices $1,...,n$.
We create another set $A$ (which is a subset of $\Omega$) that contains the $1,...,n$ with
probabilities $p_1,...,p_n$, which are independent of each other:
\begin{eqnarray}
(p_1,...,p_n)\in [0,1]^n.
\end{eqnarray}
We repeat this experiment $N$ times (where $N$ is large enough for statistics to make sense), and we get a collection of sets $A$, which we denote as ${\cal I}(p_1,...,p_n)$.
For each $A$ we also consider the corresponding set ${\mathfrak S}(A)$ that is the union of the sets $S_i$ with indices in $A$.
We call ${\cal S}(p_1,...,p_n)$ the collection of these sets.

Every subset $A\subseteq \Omega$ belongs to ${\cal I}(p_1,...,p_n)$ with multiplicity $N{\mathfrak p}(A)$, where ${\mathfrak p}(A)$ is given below in Eq.(\ref{27}).
The collection ${\cal I}(p_1,...,p_n)$ can be viewed as a multiset (which is a generalization of the concept of set, that allows the same element $A$ to appear many times).

In the large $N$ limit, and with the multiplicity taken into account:
\begin{itemize}
\item
The percentage of sets $A$ in ${\cal I}(p_1,...,p_n)$  or equivalently the percentage of sets ${\mathfrak S}(A)$ in ${\cal S}(p_1,...,p_n)$,
is: 
\begin{eqnarray}\label{27}
{\mathfrak p}(A)=\prod _{i\in A}p_i\prod _{j\in {\overline A}}(1-p_j);\;\;\;\sum _{A\subseteq \Omega}{\mathfrak p}(A)=1.
\end{eqnarray}
Here $\overline A=\Omega \setminus A$ is the complement of the set $A$. 
In the product we include the probabilities $p_i$ for the indices that belong to $A$, and the probabilities $1-p_j$ for the indices that do not belong to $A$.
Special cases are:
\begin{itemize}
\item
If $A=\emptyset$ or $A=\Omega$ we get
\begin{eqnarray}\label{27a}
{\mathfrak p}(\emptyset)=(1-p_1)...(1-p_n);\;\;\;{\mathfrak p}(\Omega)=p_1...p_n
\end{eqnarray}
\item
If $p_i=1$ when $i\in A$ and $p_i=0$ when $ i\notin A$,
then
\begin{eqnarray}\label{AB1}
{\mathfrak p}(B)=\delta (B,A).
\end{eqnarray}
In particular if $p_1=...=p_n=1$, then 
\begin{eqnarray}\label{AB2}
{\mathfrak p}(B)=\delta (B,\Omega).
\end{eqnarray}
\item
If $p_1=...=p_n=p$ then
\begin{eqnarray}\label{AB4}
{\mathfrak p}(A)=p^a(1-p)^{n-a};\;\;\;a=|A|.
\end{eqnarray}
For small values of the probability $p$ the ${\mathfrak p}(A)$ are large, for small sets $A$ (small $a$).
For large values of the probability $p$ the ${\mathfrak p}(A)$ are large, for large sets $A$ (large $a$).
We note that ${\mathfrak p}(A)$ remains invariant under the transformations
\begin{eqnarray}
p\;\rightarrow\;1-p^{\prime};\;\;\;a\;\rightarrow\;n-a^{\prime}.
\end{eqnarray}
\end{itemize}

\item
Let $A$ be a given subset of $\Omega$. The percentage of sets in ${\cal I}(p_1,...,p_n)$ that are supersets of $A$
or equivalently the percentage of sets in ${\cal S}(p_1,...,p_n)$ that are supersets of ${\mathfrak S}(A)$, is: 
\begin{eqnarray}
{ P}(A)=\prod _{i\in A}p_i.
\end{eqnarray}
In the product we include the probabilities $p_i$ for the indices that belong to $A$.
The indices that do not belong to $A$ might or might not belong to the set, and for this reason their probabilities are not included in the product.
$P(A)$ are cumulative probabilities in the sense that
\begin{eqnarray}\label{456}
{ P}(A)=\sum _{B\supseteq A}{\mathfrak p}(B).
\end{eqnarray}
Special cases are:
\begin{itemize}
\item
If $A=\emptyset$ or $A=\Omega$ we get
\begin{eqnarray}
{ P}(\emptyset)=1;\;\;\;P(\{i\})=p_i;\;\;\;P(\Omega)=p_1...p_n.
\end{eqnarray}
\item
If $p_i=1$ when $i\in A$ and $p_i=0$ when $ i\notin A$,
then
\begin{eqnarray}
&&P(B)=1\;{\rm when}\;B\subseteq A\nonumber\\
&&P(B)=0\;{\rm otherwise}.
\end{eqnarray}
In particular if $p_1=...=p_n=1$, then  $P(B)=1$ for all $B$.
\item
If $p_1=...=p_n=p$ then
\begin{eqnarray}
P(A)=p^{|A|}.
\end{eqnarray}
\end{itemize}

\item
If $A$ is a given subset of $\Omega$, the percentage of sets in ${\cal I}(p_1,...,p_n)$ that contain no elements of $A$,
or equivalently the percentage of sets in ${\cal S}(p_1,...,p_n)$ that contain no elements of ${\mathfrak S}(A)$, is:
\begin{eqnarray}\label{a1}
P(\neg A)=\prod _{i\in A}(1-p_i).
\end{eqnarray}
Special cases are:
\begin{eqnarray}
{ P}(\neg \emptyset)=1;\;\;\;P(\neg \Omega)=(1-p_1)...(1-p_n).
\end{eqnarray}
We note that
\begin{eqnarray}\label{44}
{\mathfrak p}(A)=P(A)P(\neg {\overline A}).
\end{eqnarray}
\end{itemize}
We note that the ${\mathfrak p}(A), P(A)$ are special cases of the $q(A), Q(A)$ correspondingly, introduced in section \ref{SEC1}.

\begin{definition}\label{def100}
\begin{itemize}
\mbox{}
\item[(1)]
A random set ${\mathfrak I}(p_1,...,p_n)$ of indices is the $2^n$ subsets $A$ of $\Omega$, with the probabilities ${\mathfrak p}(A)$ in Eq.(\ref{27}) attached to each of them.
\begin{eqnarray}
{\mathfrak I}(p_1,...,p_n)=\{(A, {\mathfrak p}(A))\;|\;A\subseteq \Omega\});\;\;\;\sum _{A\subseteq \Omega}{\mathfrak p}(A)=1.
\end{eqnarray}
The $2^n$ probabilities ${\mathfrak p}(A)$ are polynomials of the $n$ probabilities $p_1,...,p_n$, and in this sense we have $n$ degrees of freedom (not $2^n$).
\item[(2)]
Given the sets $S_1,...,S_n$, the random set ${\mathfrak S}(p_1,...,p_n)$, is the $2^n$ sets ${\mathfrak S}(A)$ in Eq.(\ref{1}) (for all subsets $A$ of $\Omega$), with the probabilities ${\mathfrak p}(A)$ in Eq.(\ref{27}) attached to each of them:
\begin{eqnarray}
{\mathfrak S}(p_1,...,p_n)=\{({\mathfrak S}(A), {\mathfrak p}(A))\;|\;A\subseteq \Omega\});\;\;\;\sum _{A\subseteq \Omega}{\mathfrak p}(A)=1.
\end{eqnarray}
\end{itemize}
\end{definition}

\begin{remark}
Related concept to our random set is the fuzzy set \cite{Z,Z1} where each element has a degree of membership (the analogue of our probabilities $p_i$).
This has been used extensively in the context of artificial intelligence.
Similar concept with the name ideal sets is used in ref\cite{G5} (page 142) in the context of games with a continuum of players.
Our definition of random sets is easily generalised to random projectors later, which can be implemented experimentally in the context of quantum theory.
\end{remark}

\begin{proposition}
\mbox{}
\begin{itemize}
\item[(1)]
\begin{eqnarray}\label{300}
\sum _{A\ni i}{\mathfrak p}(A)=p_i.
\end{eqnarray}
In the large $N$ limit, the percentage of the index $i$  in all sets in ${\cal I}(p_1,...,p_n)$, is $p_i$.
\item[(2)]
\begin{eqnarray}\label{78}
&&P(A)P(B)=P(A\cup B)P(A\cap B)\nonumber\\
&&P(\neg A)P(\neg B)=P[\neg (A\cup B)]P[\neg (A\cap B)]\nonumber\\
&&{\mathfrak p}(A){\mathfrak p}(B)={\mathfrak p}(A\cup B){\mathfrak p}(A\cap B).
\end{eqnarray}
\item[(3)]
\begin{eqnarray}\label{a2}
P(\neg A)=1+\sum _{B\subseteq A}(-1)^{|B|}P(B)=1-\sum _{i\in A}p_i+\sum _{i,j\in A} p_ip_j-\sum _{i,j,k\in A}p_ip_jp_k+....
\end{eqnarray}
\item[(4)]
If $A\subseteq B$ then
\begin{eqnarray}
P(A)\ge P(B);\;\;\;P(\neg A)\ge P(\neg B).
\end{eqnarray}
\item[(5)]
The derivatives of ${\mathfrak p}(A)$ are:
\begin{eqnarray}\label{a46}
&&{\rm if}\;i\in A\;{\rm then}\;\frac{\partial {\mathfrak p}(A)}{\partial p_i}=\frac{{\mathfrak p}(A)}{p_i}\ge 0\nonumber\\
&&{\rm if}\;i\notin A\;{\rm then}\;\frac{\partial {\mathfrak p}(A)}{\partial p_i}=-\frac{{\mathfrak p}(A)}{1-p_i}\le 0.
\end{eqnarray}
The values of these partial derivatives at the point $p_1=...=p_n=p$, are
\begin{eqnarray}\label{478}
&&{\rm if}\;i\in A\;{\rm then}\;;\frac{\partial {\mathfrak p}(A)}{\partial p_i}\left |\right ._{p_1=...=p_n=p}=p^{a-1}(1-p)^{n-a};\;\;\;a=|A|\nonumber\\
&&{\rm if}\;i\notin A\;{\rm then}\;;\frac{\partial {\mathfrak p}(A)}{\partial p_i}\left |\right ._{p_1=...=p_n=p}=-p^a(1-p)^{n-a-1}.
\end{eqnarray}
\end{itemize}
\end{proposition}
\begin{proof}
\begin{itemize}
\item[(1)]
We express the subsets of $\Omega$ that contain $i$ as
\begin{eqnarray}
A=\{i\}\cup A_1;\;\;\;A_1\subseteq \Omega _1=\Omega \setminus \{i\}.
\end{eqnarray}
Then ${\mathfrak p}(A)=p_i{\mathfrak p}(A_1)$ and 
\begin{eqnarray}
\sum _{A\ni i}{\mathfrak p}(A)=p_i\sum _{A_1\subseteq \Omega _1}{\mathfrak p}(A_1)=p_i.
\end{eqnarray}
We used here Eq.(\ref{27}) for the set $\Omega _1$.
\item[(2)]
The proof of the first two relations in Eq.(\ref{78}) is straightforward.
Using the second of these relations and the fact that
\begin{eqnarray}
\overline {A\cup B}=\overline A\cap \overline B;\;\;\;\overline {A\cup B}=\overline A\cap \overline B
\end{eqnarray}
we prove that
\begin{eqnarray}
P(\neg \overline A)P(\neg \overline B)=P[\neg \overline {(A\cup B)}]P[\neg \overline{(A\cap B)}]
\end{eqnarray}
We then use Eq.(\ref{44}) to prove the third relation in Eq.(\ref{78}).
\item[(3)]
We perform the multiplication in the right hand side of Eq.(\ref{a1}) and we get the terms shown in Eq.(\ref{a2}).
\item[(4)]
If $A\subseteq B$ it is straightforward to prove that $P(A)\ge P(B)$ and $P(\neg A)\ge P(\neg B)$.
\item[(5)]
This is proved with straightforward differentiation.
\end{itemize}
\end{proof}

The average cardinality of the random set ${\mathfrak I}(p_1,...,p_n)$ is
\begin{eqnarray}
\widehat {\mathfrak I}(p_1,...,p_n) =\sum _{A\subseteq \Omega}|A|{\mathfrak p}(A)=p_1+...+p_n.
\end{eqnarray}
This is because the element $i$ is included with probability $p_i$. 

\begin{proposition}\label{PRO2}
\begin{itemize}
\item[(1)]
The average cardinality of the random set ${\mathfrak S}(p_1,...,p_n)$ is
\begin{eqnarray}\label{a45}
\widehat {\mathfrak S}(p_1,...,p_n) =\sum _{A\subseteq \Omega}\mu (A){\mathfrak p}(A)=\sum _{A\subseteq \Omega}{\mathfrak d} (A)P(A).
\end{eqnarray}
\item[(2)]
In the special case that $p_i=1$ when $i\in A$ and $p_i=0$ when $ i\notin A$, we get $\widehat {\mathfrak S}(p_1,...,p_n) =\mu (A)$.
In particular $\widehat {\mathfrak S}(1,...,1)=\mu(\Omega)$ and $\widehat {\mathfrak S}(0,...,0)=0$.
\item[(3)]
In the special case of non-overlapping sets ($S_i\cap S_j=\emptyset$ for all $i,j$)
\begin{eqnarray}\label{v4}
\widehat {\mathfrak S}(p_1,...,p_n) =\sum _{i}\mu (\{i\})p_i.
\end{eqnarray}
\end{itemize}
\item[(4)]
In the special case that $p_1=...=p_n=p$, we get
\begin{eqnarray}\label{fe}
\widehat {\mathfrak S}(p,...,p) =\sum _{a=1}^n\mu _a p^a(1-p)^{n-a};\;\;\;\mu_a=\sum _{|A|=a}\mu (A)
\end{eqnarray}
\end{proposition}
\begin{proof}
\begin{itemize}
\item[(1)]
We use lemma \ref{PRO1} with $q(A)\;\rightarrow \;{\mathfrak p}(A)$. Then the definitions in Eq.(\ref{45A}) give
\begin{eqnarray}
Q(A)\;\rightarrow \;P(A);\;\;\;{\mathfrak Q}_i\;\rightarrow \;p,
\end{eqnarray}
as it is seen from Eqs.(\ref{456}), (\ref{300}).
\item[(2)]
This is proved using Eqs(\ref{AB1}), (\ref{AB2}).
\item[(3)]
This is proved using Eq.(\ref{X}).
\item[(4)]
This is proved using Eq.(\ref{AB4}).
\end{itemize}
\end{proof}
\begin{remark}
In Eq.(\ref{fe}) for small values of the probability $p$, the small sets $A$ (small $a$) contribute most.
For large values of the probability $p$, the large sets $A$ (large $a$) contribute most.
\end{remark}

\section{A continuous partition of the total cardinality of a finite number of overlapping sets}

\subsection{An infinite number of cardinality densities}\label{sec25}
Let $m$ be a `mass' distributed along the curve ${\cal C}=(x_1(t),...,x_n(t))$ in the $n$-dimensional $x$-space.
The 
\begin{eqnarray}
\frac{dm}{dt}=\sum _i\frac{\partial m}{\partial x_i}\frac{dx_i}{dt},
\end{eqnarray}
is mass density along this curve, and its line integral gives the total value of $m$.
In our context we give the following definition.
\begin{definition}
The density of the average cardinality  $\widehat {\mathfrak S}(p_1,...,p_n)$ along a curve
${\cal C}=(p_1(t),...,p_n(t))$ in the $n$-dimensional probability space, is
\begin{eqnarray}
\frac{d\widehat {\mathfrak S}}{dt}=\sum _i\frac{\partial \widehat {\mathfrak S}}{\partial p_i}\frac{dp_i}{dt}.
\end{eqnarray}
\end{definition} 
\begin{lemma}\label{LE1}
The $\frac{\partial \widehat {\mathfrak S}}{\partial p_i}$ is given by
 \begin{eqnarray}\label{11A}
 \frac{\partial \widehat {\mathfrak S}}{\partial p_i}=
 \sum _{A\ni i}\left [\mu(A)-\mu (A\setminus \{i\})\right ]\frac{{\mathfrak p}(A)}{p_i}\ge 0.
 \end{eqnarray}
 In the special case $p_1=...=p_n=p$ this reduces to the following polynomial of order $n-1$:
 \begin{eqnarray}\label{a34}
  \frac{\partial \widehat {\mathfrak S}}{\partial p_i}=\sum _{a=1}^n\lambda _{a}(i)p^{a-1}(1-p)^{n-a}
 \end{eqnarray}
 The $ \lambda _{a}(i)$ have been defined in Eq.(\ref{b90}).
\end{lemma}
\begin{proof}
From Eqs(\ref{a45}),(\ref{a46}) we get
\begin{eqnarray}\label{111}
 \frac{\partial \widehat {\mathfrak S}}{\partial p_i}=
 \sum _{A\ni i}\mu(A)\frac{{\mathfrak p}(A)}{p_i}- \sum _{A\not \ni i}\mu (A)\frac{{\mathfrak p}(A)}{1-p_i}
 \end{eqnarray}
There is a bijective map between the sets $A$ that contains the element $i$, and the sets $A\setminus\{i\}$.
Therefore Eq.(\ref{111}) is equivalent to
 \begin{eqnarray}
 \frac{\partial \widehat {\mathfrak S}}{\partial p_i}=
 \sum _{A\ni i}\left [\mu(A)\frac{{\mathfrak p}(A)}{p_i}-\mu (A\setminus \{i\})\frac{{\mathfrak p}(A\setminus \{i\})}{1-p_i}\right ].
 \end{eqnarray}
 We then use the fact that
  \begin{eqnarray}
\frac{{\mathfrak p}(A)}{p_i}=\frac{{\mathfrak p}(A\setminus \{i\})}{1-p_i},
 \end{eqnarray}
and we get the expression in Eq.(\ref{11A}) which is easily seen to be non-negative.

In the special case that the probabilities are equal to each other, using Eq.(\ref{478}) we get the expression in Eq.(\ref{a34}). 
\end{proof}

\begin{lemma}
Let
 \begin{eqnarray}
r_a=B(a,n-a+1)\sum _{i=1}^n\lambda _{a}(i)=\sum _{(A, A_-)} B(a,n-a+1) [\mu(A)-\mu (A_-)]
\end{eqnarray}
The last summation is over all pairs of sets $(A, A_-)$ with $|A|=a$, $A_-\subset A$ and $|A_-|=a-1$.
Then
\begin{eqnarray}\label{LLL}
\sum _{a=1}^nr_a=\mu (\Omega).
\end{eqnarray}
\end{lemma}
\begin{proof}
Eq.(\ref{LLL}) follows from Eq.(\ref{b90}).
\end{proof}
$r_a$ gives the average increment in the cardinality as we go from a union of $(a-1)$ sets to a union of $a$ sets.
For example
\begin{eqnarray}
r_1=\frac{1}{n}\sum _{i=1}^n\mu(\{i\});\;\;\;r_2=\frac{2}{n(n-1)}\sum _{i,j}\{[\mu(\{i,j\})-\mu(\{i\})]+[\mu(\{i,j\})-\mu(\{j\})]\}.
\end{eqnarray}

The following proposition introduces an infinite number of continuous partitions of the cardinality of the union of a finite number of overlapping sets.
\begin{proposition}\label{PRO10}
Let  ${\cal C}$ be a differentiable monotonically increasing curve in the probability hypercube $[0,1]^n$ that joins the point $(0,...,0)$ with $(1,...,1)$:
\begin{eqnarray}\label{curve}
{\cal C}=\{p_i(t)\;|\;0\le t\le 1;\;\;p_i(0)=0;\;\;p_i(1)=1;\;\;\frac{dp_i}{dt}\ge 0\}.
\end{eqnarray}
We use the notation
\begin{eqnarray}
\frac{\partial \widehat {\mathfrak S}}{\partial p_i}({\cal C})=\frac{\partial \widehat {\mathfrak S}}{\partial p_i}[p_1(t),...,p_n(t)]
\end{eqnarray}
 Also let ${\cal D}$ be the diagonal
 \begin{eqnarray}\label{diag}
{\cal D}=\{p_i(t)=t\;|\;0\le t\le 1\}
\end{eqnarray}
 \begin{itemize}
\item[(1)]
 The cardinality $\mu (\Omega)$ of the union of all sets with indices in $\Omega$, can be written as
\begin{eqnarray}\label{67}
\mu (\Omega)=\int _0^1dt M_\Omega (t|{\cal C});\;\;\;M_\Omega (t|{\cal C})=\sum _{i=1}^n \frac{\partial \widehat {\mathfrak S}}{\partial p_i}({\cal C}) \frac{dp_i}{dt}\ge 0
\end{eqnarray}
This is a `continuous partition' of $\mu(\Omega)$.
\item[(2)]
{\bf Diagonal property:}
 In the special case that ${\cal C}$ is the diagonal ${\cal D}$, we get
 \begin{eqnarray}\label{hhh}
&&M_\Omega (t|{\cal D})=\sum _{i=1}^n\frac{\partial \widehat {\mathfrak S}}{\partial p_i}({\cal D}) 
=\sum _{a=1}^n\sum _{i=1}^n\lambda _{a}(i)t^{a-1}(1-t)^{n-a}
=\sum _{a=1}^nr_a \frac{t^{a-1}(1-t)^{n-a}}{B(a,n-a+1)}
\nonumber\\
&& \int _0^1dt M_\Omega (t|{\cal D})=\mu (\Omega).
\end{eqnarray}

\end{itemize}
\end{proposition}
\begin{proof}
\begin{itemize}
\item[(1)]
The first of Eqs(\ref{67}) is definition of $M_\Omega (i,t|{\cal C})$.
The result is non-negative because both $\frac{\partial \widehat {\mathfrak S}}{\partial p_i}({\cal C})\ge 0$ and $\frac{dp_i}{dt}\ge 0$.
We next prove the second of Eqs.(\ref{67}). We use the relation
\begin{eqnarray}
\sum _{i=1}^n\frac{\partial \widehat {\mathfrak S}}{\partial p_i}\frac{dp_i}{dt}=\frac{d \widehat {\mathfrak S}}{d t},
\end{eqnarray}
and we get
\begin{eqnarray}
\sum _{i=1}^n\int _0^1\frac{\partial \widehat {\mathfrak S}}{\partial p_i}\frac{dp_i}{dt}dt=\widehat {\mathfrak S}(t=1)-\widehat {\mathfrak S}(t=0).
\end{eqnarray}
But $\widehat {\mathfrak S}(t=1)=\widehat {\mathfrak S}(1,...,1)=\mu(\Omega)$ and  $\widehat {\mathfrak S}(t=0)=\widehat {\mathfrak S}(0,...0)=0$.
This completes the proof.
\item[(2)]
Along the diagonal we use Eq.(\ref{a34}), the fact that $\frac{dp_i}{dt}=1$. This leads to the first of Eqs.(\ref{hhh}).

The second of Eqs.(\ref{hhh}) is a special case of Eq.(\ref{67}) which we have proved above, but we can also prove it directly using the integral 
\begin{eqnarray}\label{bbb}
 \int _0^1t^x(1-t)^ydt=B(x+1,y+1),
 \end{eqnarray}
and Eq.(\ref{LLL}).
\end{itemize}
\end{proof}
\begin{remark}
\mbox{}
\begin{itemize}

 \item[(1)]
  Each curve ${\cal C}$  leads to a different set of $\{M_\Omega (t|{\cal C})\}$.
 Therefore we get an infinite number of different ways of dividing the overlaps into the various sets (labelled by the curves ${\cal C}$).

  \item[(2)]
  The use of random sets leads to a continuous partition of the cardinality $\mu(\Omega)$ of the union of all sets with indices in $\Omega$, in Eq.(\ref{67}).
The term `continuous' refers to the fact that $t$ is a continuous variable and 
$\mu(\Omega)$ is expressed as a line integral of the  density of the average cardinality.
The term  `partition of the cardinality' refers to Eqs.(\ref{67}),(\ref{hhh}) and should not be confused with the term`partition of a set into subsets', where there is the requirement that the subsets should not overlap with each other.
 \item[(3)]

 Using the integral in Eq.(\ref{bbb})
we show that the Shapley cardinality $M_\Omega (i)$ in Eq.(\ref{b90}) can be written as the integral
\begin{eqnarray}\label{b900}
M_\Omega (i)=\sum _{a=1}^n\lambda _{a}(i)B(a,n-a+1)=\sum _{a=1}^n\lambda _{a}(i) \int _0^1dtt^{a-1}(1-t)^{n-a}=\int _0^1dt\frac{\partial \widehat {\mathfrak S}}{\partial p_i}({\cal D}) .
\end{eqnarray}
The same quantity $\lambda _{a}(i)t^{a-1}(1-t)^{n-a}$
appears in both $M_\Omega (i)$ and  $M_\Omega (t|{\cal D})$ (Eqs.(\ref{hhh}),(\ref{b900})).
But in  $M_\Omega (i)$ it is integrated over $t$ and it gives the Beta function, while in the $M_\Omega (t|{\cal D})$ it is summed over $i$. 
In section \ref{FFF} we had a clearly defined collection of sets labelled with the indices $i\in \Omega$, and we allocated to each of them the Shapley cardinality $M_\Omega (i)$.
In section \ref{GGG} we have a random collection of sets and we defined the $M_\Omega (t|{\cal D})$ in terms of the probability $t$.

\item[(4)]
The physical interpretation of the Beta function $B(a,n-a+1)$ as a probability in Eq.(\ref{b900}), has been discussed earlier in section \ref{sec20}.
We considered the probability that the set $S_i$ will join a union after the $a-1$ sets with indices in $A\setminus \{i\}$ and before the $n-a$
sets with indices in $\Omega \setminus A$ (where $a=|A|$).
We modify that argument by taking each of these sets with probability $t$.  This leads to the probability density $t^{a-1}(1-t)^{n-a}$
and its integral is the probability $B(a, n-a+1)$.

 \item[(5)]
 Relations analogous to Eqs.(\ref{67}), (\ref{hhh}) in the context of cooperative game theory are given in ref\cite{owen} for the case of a finite number of players, and in ref\cite{G5} for the case of a continuum of players.

 \end{itemize}

\end{remark}

\begin{table}
\caption{Let  $S_1=\{a,b\}$, $S_2=\{b,c\}$, $S_3=\{a,c,d\}$, $S_4=\{b,d\}$. The random set ${\mathfrak I}(\frac{1}{2}, \frac{1}{3}, \frac{1}{4},\frac{1}{5})$ (i.e. the 
$(A, {\mathfrak p}(A))$ for all $A\subseteq \Omega$), and the random set ${\mathfrak S}(\frac{1}{2}, \frac{1}{3}, \frac{1}{4},\frac{1}{5})$ (i.e. the 
$({\mathfrak S}(A), {\mathfrak p}(A))$ for all $A\subseteq \Omega$) are shown.
The probabilities $P(A)$, $P(\neg A)$ are also given.}
\def\arraystretch{2}
\begin{tabular}{|c|c|c|c|c|c|}\hline
$A$&${\mathfrak S}(A)$&$\mu(A)$&${\mathfrak p}(A)$&$P(A)$&$P(\neg A)$\\\hline
$\emptyset$&$\emptyset$&$0$&$\frac{1}{5}$&$1$&$1$\\\hline
$\{1\}$&$\{a,b\}$&$2$&$\frac{1}{5}$&$\frac{1}{2}$&$\frac{1}{2}$\\\hline
$\{2\}$&$\{b,c\}$&$2$&$\frac{1}{10}$&$\frac{1}{3}$&$\frac{2}{3}$\\\hline
$\{3\}$&$\{a,c,d\}$&$3$&$\frac{1}{15}$&$\frac{1}{4}$&$\frac{3}{4}$\\\hline
$\{4\}$&$\{b,d\}$&$2$&$\frac{1}{20}$&$\frac{1}{5}$&$\frac{4}{5}$\\\hline
$\{1,2\}$&$\{a,b,c\}$&$3$&$\frac{1}{10}$&$\frac{1}{6}$&$\frac{1}{3}$\\\hline
$\{1,3\}$&$\{a,b,c,d\}$&$4$&$\frac{1}{15}$&$\frac{1}{8}$&$\frac{3}{8}$\\\hline
$\{1,4\}$&$\{a,b,d\}$&$3$&$\frac{1}{20}$&$\frac{1}{10}$&$\frac{2}{5}$\\\hline
$\{2,3\}$&$\{a,b,c,d\}$&$4$&$\frac{1}{30}$&$\frac{1}{12}$&$\frac{1}{2}$\\\hline
$\{2,4\}$&$\{b,c,d\}$&$3$&$\frac{1}{40}$&$\frac{1}{15}$&$\frac{8}{15}$\\\hline
$\{3,4\}$&$\{a,b,c,d\}$&$4$&$\frac{1}{60}$&$\frac{1}{20}$&$\frac{3}{5}$\\\hline
$\{1,2,3\}$&$\{a,b,c,d\}$&$4$&$\frac{1}{30}$&$\frac{1}{24}$&$\frac{1}{4}$\\\hline
$\{1,2,4\}$&$\{a,b,c,d\}$&$4$&$\frac{1}{40}$&$\frac{1}{30}$&$\frac{4}{15}$\\\hline
$\{1,3,4\}$&$\{a,b,c,d\}$&$4$&$\frac{1}{60}$&$\frac{1}{40}$&$\frac{3}{10}$\\\hline
$\{2,3,4\}$&$\{a,b,c,d\}$&$4$&$\frac{1}{120}$&$\frac{1}{60}$&$\frac{2}{5}$\\\hline
$\{1,2,3,4\}$&$\{a,b,c,d\}$&$4$&$\frac{1}{120}$&$\frac{1}{120}$&$\frac{1}{5}$\\\hline
\end{tabular} \label{t1}
\end{table}

\subsection{Example}\label{ex}
Let
\begin{eqnarray}
S_1=\{a,b\};\;\;\;S_2=\{b,c\};\;\;\;S_3=\{a,c,d\};\;\;\;S_4=\{b,d\}.
\end{eqnarray}
The set of indices in this case is  $\Omega =\{1,2,3,4\}$.

The random set of indices ${\mathfrak I}(\frac{1}{2}, \frac{1}{3}, \frac{1}{4},\frac{1}{5}) $ is shown in table \ref{t1} (the $2^4=16$ subsets $A$ of $\Omega$ and the corresponding probabilities ${\mathfrak p}(A)$,
$P(A)$ and $P(\neg A)$). In this case $\widehat {\mathfrak I}(\frac{1}{2}, \frac{1}{3}, \frac{1}{4},\frac{1}{5}) =1.28$.

The random set ${\mathfrak S}(\frac{1}{2}, \frac{1}{3}, \frac{1}{4},\frac{1}{5})$ is also shown (the $2^4=16$ sets ${\mathfrak S}(A)$ and the corresponding probabilities).
In this case $\widehat {\mathfrak S}(\frac{1}{2}, \frac{1}{3}, \frac{1}{4},\frac{1}{5})=2.26$.

We next consider the diagonal line ${\cal D}$ and using Eq.(\ref{a34}) and table \ref{t1} we calculate the derivatives:
 \begin{eqnarray}
 &&\frac{\partial \widehat {\mathfrak S}}{\partial p_i}({\cal D})=\lambda_{1} (i)(1-t)^3+\lambda _{2}(i)t(1-t)^2+\lambda _{3}(i)t^2(1-t)\nonumber\\
&&\lambda _{1}(1)=\lambda _{1}(2)=\lambda _{1}(4)=2;\;\;\;\lambda _{2}(1)=\lambda _{2}(2)=\lambda _{2}(4)=3;\;\;\;\lambda_{3}(1)=\lambda_{3}(2)=\lambda_{3}(4)=1\nonumber\\
&& \lambda_{1}(3)=3;\;\;\;\lambda _{2}(3)=6;\;\;\;\lambda_{3}(3)=3.
 \end{eqnarray}
Also $\lambda _{4}(i)=0$. 

In this example
 \begin{eqnarray}
r_1=\frac{9}{4};\;\;\;r_2=\frac{5}{4};\;\;\;r_3=\frac{1}{2};\;\;\;r_4=0,
 \end{eqnarray}
and
 \begin{eqnarray}
M_\Omega (t|{\cal D})=4(1-t)^3r_1+12t(1-t)^2r_2+12t^2(1-t)r_3+4t^3r_4.
\end{eqnarray}
Also $\mu(\Omega)=4$.
Therefore 
 \begin{eqnarray}
\mu(\Omega)=\int_0^1dt M_\Omega (t|{\cal D})=r_1+r_2+r_3+r_4.
\end{eqnarray}

\section{Continuous resolutions of the identity in a finite-dimensional Hilbert space}

\subsection{Random projectors to the $2^n$ subspaces spanned by a total set of $n$ vectors}
In the $d$-dimensional Hilbert space $H(d)$ we consider the projectors $\Pi(A)$ to the $2^n$ subspaces spanned by the $n$ vectors in the
 pre-basis $\Sigma$  (Eq.(\ref{sig})).
\begin{definition}\label{def200}
A random projector $\varpi (p_1,...,p_n)$ in the Hilbert space $H(d)$, 
is the $2^n$ projectors $\Pi(A)$ (for all subsets $A$ of $\Omega$), with the probabilities ${\mathfrak p}(A)$ in Eq.(\ref{27}) attached to each of them:
\begin{eqnarray}
\varpi (p_1,...,p_n)=\{(\Pi(A),{\mathfrak p}(A))\;|\;A\subseteq \Omega\});\;\;\;\sum _{A\subseteq \Omega}{\mathfrak p}(A)=1.
\end{eqnarray}
\end{definition}
We note that if $|A|\ge d$ then $\Pi(A)={\bf 1}$.
The following proposition is analogous to proposition \ref{PRO2} in the context of set theory. 
\begin{proposition}\label{PRO11}
\begin{itemize}
\item[(1)]
The average of the random projector is the positive semi-definite operator (which in general is not a projector)
\begin{eqnarray}\label{RT1}
\widehat \varpi (p_1,...,p_n)=\sum _A {\mathfrak p}(A)\Pi (A)=\sum _A P(A){\mathfrak D} (A).
\end{eqnarray}
\item[(2)]
In the special case that $p_i=1$ when $i\in A$ and $p_i=0$ when $ i\notin A$, we get $\widehat \varpi (p_1,...,p_n) =\Pi (A)$.
In particular, if this set $A$ has cardinality $|A|\ge d$ then  $\widehat \varpi (p_1,...,p_n) ={\bf 1}$.
Also
\begin{eqnarray}
\widehat \varpi(1,...,1)={\bf 1};\;\;\;\widehat \varpi (0,...,0)={\bf 0}.
\end{eqnarray}
\item[(3)]
In the special case that the pre-basis $\Sigma$ is an orthonormal basis
\begin{eqnarray}\label{er}
\widehat \varpi(p_1,...,p_d) =\sum _{i=1}^d\Pi (\{i\})p_i.
\end{eqnarray}
\item[(4)]
In the special case that $p_1=...=p_n=p$, we get
\begin{eqnarray}\label{fe1}
\widehat {\varpi}(p,...,p) =\sum _{a=1}^n\left (\sum _{|A|=a}\Pi (A)\right ) p^a(1-p)^{n-a}.
\end{eqnarray}
\end{itemize}
\end{proposition}
\begin{proof}
The proof is analogous to the proof of proposition \ref{PRO2}. 
\begin{itemize}
\item[(1)]
The proof of Eq.(\ref{RT1}) is analogous to the proof of Eq.(\ref{a45}).
\item[(2)]
The proof is analogous to the proof of the second statement of proposition \ref{PRO2}. 
\item[(3)]
The proof is analogous to the proof of the third statement of proposition \ref{PRO2}. 
The non-overlapping sets there, correspond to the orthonormal basis here.
\item[(4)]
The proof of Eq.(\ref{fe1}) is analogous to the proof of Eq.(\ref{fe}).
\end{itemize}
\end{proof}

\subsection{An infinite number of continuous resolutions of the identity}

\begin{definition}
The density of the average $\widehat \varpi (p_1,...,p_n)$ of the random projector along a curve
${\cal C}=(p_1(t),...,p_n(t))$ in the $n$-dimensional probability space, is
\begin{eqnarray}
\frac{d\widehat \varpi }{dt}=\sum _i\frac{\partial \widehat \varpi }{\partial p_i}\frac{dp_i}{dt}.
\end{eqnarray}
\end{definition}

\begin{lemma}
$\frac{\partial \widehat {\varpi}}{\partial p_i}$ is a Hermitian positive semi-definite operator, given by
 \begin{eqnarray}\label{11AA}
 \frac{\partial \widehat {\varpi}}{\partial p_i}=
 \sum _{A\ni i}\left [\Pi(A)-\Pi (A\setminus \{i\})\right ]\frac{{\mathfrak p}(A)}{p_i}.
 \end{eqnarray}
 In the special case $p_1=...=p_n=p$ this reduces to
 \begin{eqnarray}\label{a34a}
  \frac{\partial \widehat {\varpi}}{\partial p_i}&=&
 \sum _{A\ni i} [\Pi(A)-\Pi (A\setminus \{i\})]p^{a-1}(1-p)^{n-a}=\sum _{a=1}^n\Lambda _{a}(i)p^{a-1}(1-p)^{n-a}
 \end{eqnarray}
$\Lambda _{a}(i)$ are $nd$ Hermitian positive semi-definite $d\times d$ matrices, given in Eq.(\ref{rrr}).
\end{lemma}

\begin{proof}
The proof of Eq.(\ref{11AA}) is analogous to the proof in Lemma \ref{LE1}, with the cardinality $\mu(A)$ replaced here by the projector $\Pi(A)$.
We note that the $\Pi(A)-\Pi (A\setminus \{i\})$ are projectors, and that in Eq.(\ref{11AA}) we have a finite sum of projectors with positive coefficients.
Any such sum is a positive semi-definite operator, and this proves that the  $\frac{\partial \widehat {\varpi}}{\partial p_i}$ are positive semi-definite operators. 
\end{proof}

\begin{remark}
The $\Pi(A)$ project into the subspaces spanned by vectors in the pre-basis.
In Eq.(\ref{a34a}) for small values of the probability $p$, the projectors $\Pi(A)$ with small sets $A$ (small $a$) contribute most.
For large values of the probability $p$, the projectors $\Pi(A)$ with large sets $A$ (large $a$) contribute most.
\end{remark}

The following proposition introduces an infinite number of continuous resolutions of the identity, for a given pre-basis.
Each of them involves an one-dimensional continuum of Hermitian positive semi-definite operators which with appropriate normalization can become density matrices.
\begin{proposition}\label{PRO100}
Let ${\cal C}$ be a curve in the probability hypercube $[0,1]^n$  as in Eq.(\ref{curve}) .
We use the notation
\begin{eqnarray}
\frac{\partial \widehat \varpi}{\partial p_i}({\cal C})=\frac{\partial \widehat \varpi }{\partial p_i}[p_1(t),...,p_n(t)]
\end{eqnarray}
The  operators
\begin{eqnarray}\label{670}
\tau _\Omega (t|{\cal C})=\sum _{i=1}^n\frac{\partial \widehat \varpi}{\partial p_i}({\cal C})\frac{dp_i}{dt};\;\;\;
 \int _0^1dt \tau _\Omega (t|{\cal C})=\bf 1.
\end{eqnarray}
are Hermitian positive semi-definite and resolve the identity.
\end{proposition}
\begin{proof}
The proof of this is analogous to the first part of the proposition \ref{PRO10}.
The $\frac{\partial \widehat \varpi}{\partial p_i}$ are positive semi-definite operators, and since $\frac{dp_i}{dt}\ge 0$ (Eq.(\ref{curve})), the $\tau _\Omega (t;{\cal C})$ are also positive semi-definite operators.
\end{proof}

Each curve ${\cal C}$ in Eq.(\ref{curve})
leads to a different set of $\{\tau_\Omega (t|{\cal C})\}$.
 Therefore we get an infinite number of resolutions of the identity.
All of them are continuous resolutions of the identity in the sense that  the unity
is resolved into the one-dimensional continuum of operators $\tau _\Omega (t|{\cal C})$ that depend on the continuous real variable $t$.
The discrete resolution of the identity in Eq.(\ref{ppp}) in terms of $n$ matrices, is replaced here by a continuous resolution of the identity.

\subsection{The diagonal continuous resolution of the identity}
For later use we introduce the following polynomial, which is a truncated binomial expansion:
\begin{eqnarray}\label{108}
&&{\mathfrak P}_{n,k}(t)=\sum _{j=0}^k\begin{pmatrix}
n\\
j\\
\end{pmatrix}t^j(1-t)^{n-j};\;\;\;k\le n;\;\;\;0\le t\le1.
 \end{eqnarray}
Some mathematical properties of these polynomials have been studied in \cite{number}.
Using the integral in Eq.(\ref{bbb}), we prove that
\begin{eqnarray}
\int_0^1dt {\mathfrak P}_{n,k}(t)=\frac{k+1}{n+1}.
\end{eqnarray}
Examples of these polynomials are:
\begin{eqnarray}
{\mathfrak P}_{n,0}(t)=(1-t)^n;\;\;\;{\mathfrak P}_{n,n-1}(t)=1-t^n;\;\;\;{\mathfrak P}_{n,n}(t)=1.
 \end{eqnarray}

\begin{definition}
\begin{eqnarray}\label{333}
{R_a}=B(a,n-a+1)\sum _{i=1}^n\Lambda_a(i)=B(a,n-a+1)\sum _{(A,A_-)}[\Pi(A)-\Pi(A_-)];\;\;\;a=1,...,d,
\end{eqnarray}
where the 
last summation is over all pairs of sets $(A, A_-)$ with $|A|=a$, $A_-\subset A$ and $|A_-|=a-1$.

\end{definition}
The $\Pi(A)-\Pi(A_-)$ are projectors which give the `increment' as we go from the $(a-1)$-dimensional subspace $H(A_-)$ to the  $a$-dimensional space $H(A)$.
These spaces are formed from vectors in the pre-basis.
Therefore the $R_a$ is an `average increment' as we go from an $(a-1)$-dimensional subspace to an  $a$-dimensional space.
For example, 
\begin{eqnarray}
{R_1}&=&\frac{1}{n}\sum _{i=1}^n\Pi(\{i\});\;\;\;{R_2}=\frac{2}{n(n-1)}\sum _{\{i,j\}}\{[\Pi(\{i,j\})-\Pi(\{i\})]+[\Pi(\{i,j\})-\Pi(\{j\})]\}\nonumber\\
R_3&=&\frac{6}{n(n-1)(n-2)}\sum _{\{i,j,k\}}\{[\Pi(\{i,j,k\})-\Pi(\{i,j\})]+[\Pi(\{i,j,k\})-\Pi(\{i,k\})]\nonumber\\&+&[\Pi(\{i,j,k\})-\Pi(\{j,k\})]\}
\end{eqnarray}
The second sum is over all $\frac{n(n-1)}{2}$ pairs of indices in a set of $n$ indices. 
The third sum is over all $\frac{n(n-1)(n-2)}{6}$ triplets of indices in a set of $n$ indices.

The $R_a$ is  not a projector, but the following lemma shows that it is a density matrix.
Eq.(\ref{reso23}) below provides a partition of unity in terms of them.

\begin{lemma}
The $R_a$ are $d$ density matrices (which for a given pre-basis are fixed) and 
\begin{eqnarray}\label{reso23}
\sum _{a=1}^d{R_a}={\bf 1}.
\end{eqnarray}
\end{lemma}
\begin{proof}
The $\Lambda_a(i)$ are Hermitian positive semi-definite matrices and therefore the $R_a$ are also Hermitian positive semi-definite matrices.
Using Eq.(\ref{rrr}) we prove that 
\begin{eqnarray}\label{pppp}
{\rm Tr} \left [\sum _{i=1}^n{\Lambda} _{a}(i)\right ]=\frac{1}{B(a,n-a+1)}.
\end{eqnarray}
Therefore the $R_a$ have trace $1$, and they are density matrices. Eq.(\ref{reso23}) follows from Eqs.(\ref{ppp}), (\ref{rrr}).
\end{proof}

We note that summation of the $B(a,n-a+1)\Lambda_a(i)$ over $i$ gives the $d$ density matrices $R_a$ in Eq.(\ref{333}), while summation over $a$ gives the $n$ density matrices $\sigma _\Omega (i)$ in Eq.(\ref{ppp}).

\begin{proposition}\label{PRO200}
Let ${\cal D}$ be the diagonal in the probability hypercube $[0,1]^n$  as in Eq.(\ref{diag}). 
\begin{itemize}
\item[(1)]
The operators
\begin{eqnarray}\label{678}
\tau _\Omega (t|{\cal D})=\sum _{i=1}^n\frac{\partial \widehat \varpi}{\partial p_i}({\cal D}) =\sum _{a=1}^d\sum _{i=1}^n\Lambda_a(i)t^{a-1}(1-t)^{n-a}
=\sum _{a=1}^d\frac{R_a}{B(a,n-a+1)}
t^{a-1}(1-t)^{n-a}
\end{eqnarray}
are Hermitian positive semi-definite and resolve the identity.
\begin{eqnarray}\label{reso10}
\int _0^1dt \tau _\Omega (t|{\cal D})={\bf 1}.
\end{eqnarray}

\item[(2)]
The trace of  $\tau _\Omega (t|{\cal D})$ is:
\begin{eqnarray}
{\rm Tr}[\tau _\Omega (t|{\cal D})]=n{\mathfrak P}_{n-1,d-1}(t).
\end{eqnarray}
\end{itemize}
\end{proposition}
\begin{proof}
\begin{itemize}
\item[(1)]
This is a corollary of proposition \ref{PRO100}. In Eq.(\ref{678}) we used Eq.(\ref{a34a}) and that
in the diagonal case $\frac{dp_i}{dt}=1$.
A direct proof of Eq.(\ref{reso10}) is 
\begin{eqnarray}
\int _0^1dt \tau _\Omega (t|{\cal D})=\sum _{a=1}^d\frac{R _a}{B(a,n-a+1)}\int _0^1dtt^{a-1}(1-t)^{n-a}=\sum _{a=1}^dR_a={\bf 1}.
\end{eqnarray}
\item[(2)]
Using Eq.(\ref{108}) we find that the trace of  $\tau _\Omega (t|{\cal D})$ is:
\begin{eqnarray}
{\rm Tr}[\tau _\Omega (t|{\cal D})]=\sum _{a=1}^d
\frac{t^{a-1}(1-t)^{n-a}}{B(a,n-a+1)}=n{\mathfrak P}_{n-1,d-1}(t).
\end{eqnarray}
\end{itemize}
\end{proof}

\subsection{The diagonal continuous resolutions of the identity versus the discrete resolutions of the identity}

The following proposition will be used to link the discrete resolution of the identity in Eq.(\ref{ppp}) with the continuous one in Eq.(\ref{reso10}). 
\begin{proposition}
The density matrices of Eq.(\ref{ppp}) can be written as integrals as follows:
\begin{eqnarray}\label{116}
\sigma _{\Omega}(i)=\frac{n}{d}\int _0^1dt \left [\sum _{a=1}^d{\Lambda _a}(i)t^{a-1}(1-t)^{n-a}\right ];\;\;\;\frac{d}{n}\sum _{i=1}^n \sigma _{\Omega}(i)={\bf 1}.
\end{eqnarray}
\end{proposition}
\begin{proof}
Using the integral in Eq.(\ref{bbb}), we get
\begin{eqnarray}
\int _0^1dt \left [\sum _{a=1}^d{\Lambda _a}(i)t^{a-1}(1-t)^{n-a}\right ]=
 \sum _{a=1}^n \Lambda_{a}(i)B(a, n-a+1)=\frac{d}{n}\sigma _\Omega (i).
 \end{eqnarray}
 Eq.(\ref{ppp}) has been used in the last step.
\end{proof}
It seen that the same quantity ${\Lambda _a}(i)t^{a-1}(1-t)^{n-a}$ appears in both $\tau _\Omega (t|{\cal D})$  (Eq.(\ref{678})) and $\sigma _{\Omega}(i)$ (Eq.(\ref{116})).
But the $\tau _\Omega (t|{\cal D})$ performs summation over $i$, while the $\sigma _{\Omega}(i)$ 
performs integration over $t$ and this gives the Beta function $B(a, n-a+1)$.
In section \ref{sec25} we interpreted the $B(a, n-a+1)$ as a probability and the $t^{a-1}(1-t)^{n-a}$ as a probability density, in the context of set theory.
The same interpretation applies here also.

\section{The ${\cal T}$-representation}

The variable $t$ has a nice interpretation as probability, but it is convinient to change variables from $t$ to $x$, where
\begin{eqnarray}
x=\frac{t}{1-t};\;\;\;t=\frac{x}{1+x};\;\;\;dt=\frac{dx}{(1+x)^2}
\end{eqnarray}
When $t$ takes values in $(0,1)$ the `probability related' variable $x$ takes values in the `probabilistic semi-axis' $[0,\infty)$. In particular $t=0.5$ corresponds to $x=1$.
Now the resolution of the identity in Eq.(\ref{reso10}) (for the diagonal case) can be rewritten as
\begin{eqnarray}\label{cd}
\int _0^\infty  dx{\cal T}(x)={\bf 1},
\end{eqnarray}
where
\begin{eqnarray}
&&{\cal T}(x)=\frac{1}{(1+x)^{n+1}}\sum _{a=1}^d x^{a-1}\frac{R_a}{B(a,n-a+1)}=\frac{1}{(1+x)^{n+1}}\sum _{a=1}^d x^{a-1}\left [\sum  _{(A, A_-)}\Pi(A)-\Pi(A_-)\right ]\nonumber\\
&&{\rm Tr}[{\cal T}(x)]=\frac{n}{(1+x)^2}{\mathfrak P}_{n-1,d-1}\left(\frac{x}{1+x}\right ).
\end{eqnarray}
The last summation is over all pairs of sets $(A, A_-)$ with $|A|=a$, $A_-\subset A$ and $|A_-|=a-1$.
We can prove Eq.(\ref{cd}) directly, using Eq.(\ref{reso23}) and the integral
\begin{eqnarray}\label{A3}
\int _0^{\infty} dx \frac{x^{a-1}}{(1+x)^{n+1}}=B(a,n-a+1).
\end{eqnarray}

We note that
\begin{itemize}
\item
Since $x$ takes positive values, ${\cal T}(x)$ is a $d\times d$ Hermitian positive semi-definite matrix, which is not a projector.
${\cal T}(x)$ is really the  $\tau _\Omega (t|{\cal D})$ in terms of a different variable.
Its elements are  ratios of polynomials of $x$ of order $d-1$ over $(1+x)^{n+1}$.
\item
${\cal T}(x)$ is related to the subspaces spanned by the vectors in the pre-basis.
For small (large) values of the probability related variable $x$, it is mainly related to the subspaces with small (large) dimension.

\end{itemize}

\subsection{Representation of states in terms of a continuum of components}
The resolution of the identity in Eq.(\ref{cd}) can be used to partition an arbitrary state $\ket{s}$ 
in the finite-dimensional Hilbert space $H(d)$,
in terms of a continuum of components ($x\in [0,\infty )$) which are not normalised, are  not orthogonal and are not independent from each other:
\begin{eqnarray}\label{cdb}
\ket{s}=\int _0^{\infty}dx\ket{s(x)};\;\;\;\ket{s(x)}={\cal T}(x) \ket{s}=\frac{1}{(1+x)^{n+1}}\sum _{a=1}^d x^{a-1}\frac{R_a\ket{s}}{B(a,n-a+1)}.
\end{eqnarray}
Each $\ket{s(x)}$ is a $d\times 1$ column whose components are ratios of polynomials of $x$ of order $d-1$ over $(1+x)^{n+1}$. 
The coefficient of $\frac{x^{a-1}}{(1+x)^{n+1}}$ is associated with the matrix $R_a$ which as we explained earlier gives the average increment as we go from 
an $(a-1)$-dimensional subspace to an $a$-dimensional space (in the spaces spanned by the vectors in the pre-basis). Also
\begin{eqnarray}
\ket{s(0)}= \frac{1}{n}R_1\ket{s}= \frac{1}{n^2}\sum _{i=1}^n\Pi(i)\ket{s};\;\;\;\ket{s(x)}\sim \frac{1}{x^{n-d+2}}\frac{R_d\ket{s}}{B(d,n-d+1)}\;\;{\rm as}\;\;x\rightarrow \infty.
\end{eqnarray}
Similar comment applies to  $\bra{s(x)}$ which is $1\times d$ row.

The scalar product of the bra state $\bra{u}$ with the ket state $\ket{s}$ is written as the double integral:
\begin{eqnarray}
\int _0^{\infty}\int _0^{\infty}dx_1dx_2\langle{u(x_1)}\ket{s(x_2)}=\langle u\ket{s}.
\end{eqnarray}

\subsection{Representation of operators}

An operator $\Theta$ is represented as follows: 
\begin{eqnarray}\label{xcv}
{\cal O}(x_1,x_2; \Theta)= {\cal T}(x_1)\Theta {\cal T}(x_2);\;\;\;\int _0^\infty dx_1\int _0^\infty dx_2{\cal O}(x_1,x_2; \Theta)=\Theta .
\end{eqnarray}
\begin{proposition}
The ket $\Theta \ket{s}$ and the bra $\bra {s}\Theta$ are represented by 
\begin{eqnarray}\label{132}
&&\Theta \ket{s}\;\rightarrow\;\ket{\Theta s (x_1)}=\int dx_2dx_3{\cal O}(x_1,x_2;\Theta)\ket{s(x_3)}\nonumber\\
&&\bra{s}\Theta \;\rightarrow\;\bra{s\Theta (x_3)}=\int dx_1dx_2\bra{s(x_1)}{\cal O}(x_2,x_3;\Theta).
\end{eqnarray}
\end{proposition}
\begin{proof}
We get
\begin{eqnarray}
\int dx_2dx_3{\cal O}(x_1,x_2;\Theta)\ket{s(x_3)}=\int dx_2x_3{\cal T}(x_1)\Theta {\cal T}(x_2){\cal T}(x_3)\ket{s}={\cal T}(x_1)\Theta \ket{s}=\ket{\Theta s (x_1)}
\end{eqnarray}
The second relation is proved in a similar way.
\end{proof}
The unity operator is represented by
\begin{eqnarray}
{\cal O}(x_1,x_2; {\bf 1})= {\cal T}(x_1) {\cal T}(x_2)=\frac{1}{[(1+x_1)(1+x_2)]^{n+1}}\sum _{a=1}^d\sum _{b=1}^d x_1^{a-1}x_2^{b-1}\frac{R_aR_b}{B(a,n-a+1)B(b,n-b+1)} .
\end{eqnarray}
If $\Theta$ is a Hermitian operator ($\Theta^\dagger=\Theta$) then
 \begin{eqnarray}
[{\cal O}(x_1,x_2; \Theta )]^\dagger= {\cal O}(x_2,x_1; \Theta)
\end{eqnarray}

\subsection{The function $F(x_1,x_2)$}

\begin{definition}
If $\rho$ is the density matrix of a system, then
\begin{eqnarray}\label{A2}
F(x_1,x_2)&=&{\rm Tr}[{\cal O}(x_1,x_2; \rho)]={\rm Tr}[\rho {\cal T}(x_2){\cal T}(x_1)]\nonumber\\&=& \frac{1}{[(1+x_1)(1+x_2)]^{n+1}}\sum _{a=1}^d\sum _{b=1}^d x_1^{a-1}x_2^{b-1}\frac{{\rm Tr}(R_a\rho R_b)}{B(a,n-a+1)B(b,n-b+1)} .
\end{eqnarray}
$(x_1,x_2)$ takes values in the `probabilistic  quadrant' $[0, \infty)\times [0,\infty)$.
\end{definition}

$F(x_1,x_2)$ takes complex values, and since $\rho$ is Hermitian, $F(x_2,x_1)=[F(x_1,x_2)]^*$.
In the case of a pure state, i.e., $\rho=\ket{s}\bra {s}$, we get
\begin{eqnarray}
F(x_1,x_2)=\langle s (x_2)\ket {s(x_1)}.
\end{eqnarray}
The expansion in Eq.(\ref{cdb}) can be rewritten in terms of normalized vectors as
\begin{eqnarray}\label{cdb1}
\ket{s}=\int _0^{\infty}dx{\sqrt {F(x,x)}} \ket{{\cal S}(x)};\;\;\; \ket{{\cal S}(x)}=\frac{1}{\sqrt {F(x,x)}}\ket{s(x)}.
\end{eqnarray}
We note that
\begin{eqnarray}\label{XX1}
\int   _0^\infty d\beta {\cal T}(\alpha){\cal T}(\beta)=\int  _0^\infty d\beta{\cal T}(\beta){\cal T}(\alpha)={\cal T}(\alpha);\;\;\;\int _0^\infty \int _0^\infty dx_1dx_2{\cal T}(x_2){\cal T}(x_1)={\bf 1}.
\end{eqnarray}
 ${\cal T}(x_2)$, ${\cal T}(x_1)$ are Hermitian positive semi-definite operators, but they do not commute and the  ${\cal T}(x_2){\cal T}(x_1)$ is not a Hermitian operator, in general.

We define the `marginal function'
\begin{eqnarray}
{\mathfrak F}(\alpha)&=&\int _0^\infty dx_1F(x_1,\alpha)= \int _0^\infty dx_2F(\alpha, x_2)={\rm Tr}[\rho {\cal T}(\alpha)]\nonumber\\&=&\frac{1}{(1+\alpha)^{n+1}}\sum _{b=1}^d \alpha^{b-1}\frac{{\rm Tr}(\rho R_b)}{B(b,n-b+1)}\ge 0.
\end{eqnarray}
The two marginal functions are exactly the same function.
The marginal  function ${\mathfrak F}(\alpha)$ is the average value of measurements with the Hermitian positive semi-definite operator ${\cal T}(\alpha)$  on an ensemble described with the density matrix $\rho$. 
But we can {\bf not} interpret the  ${\mathfrak F}(\alpha)$ as a probability distribution because the  ${\cal T}(\alpha)$,  ${\cal T}(\alpha ^\prime )$ do not commute.

\begin{proposition}
The moments of the function $F(x_1,x_2)$ are
\begin{eqnarray}\label{A1}
&&\langle x_1^\mu x_2^\nu\rangle=\int _0^\infty dx_1 \int _0^\infty dx_2  x_1^\mu x_2^\nu F(x_1,x_2)=\sum _{a,b}{\rm Tr}(R_a\rho R_b)\frac{(a)_\mu }{(n-a+1-\mu)_\mu}\frac{(b)_\nu }{(n-b+1-\nu)_\nu}\nonumber\\
&&(e)_\mu=\frac{\Gamma (e+\mu)}{\Gamma (e)}.
\end{eqnarray}
In particular for $\mu=\nu=0$ we get
\begin{eqnarray}\label{A4}
\int _0^\infty dx_1\int _0^\infty dx_2 F(x_1,x_2)=\int _0^\infty d\alpha {\mathfrak F}(\alpha)=1.
\end{eqnarray}
For $\mu=1$, $\nu=0$ we get
\begin{eqnarray}\label{A11}
\langle x_1\rangle =\int _0^\infty dx_1\int _0^\infty dx_2 x_1  F(x_1,x_2)=\sum _{a=1}^d {\rm Tr}(\rho R_a)\frac{a }{n-a},
\end{eqnarray}
and similarly for $\mu=0$, $\nu=1$.
\end{proposition}
\begin{proof}
Eq.(\ref{A1}) is proved using Eq.(\ref{A2}) and the integral in Eq.(\ref{A3}).
We also use the relation
\begin{eqnarray}
\frac{B(a+\nu,n-a+1-\nu)}{B(a,n-a+1)}=\frac{(a)_\nu }{(n-a+1-\nu)_\nu}
\end{eqnarray}
For Eq.(\ref{A4}), we substitute $\mu=\nu=0$ in 
Eq.(\ref{A1}),  and then use Eq.(\ref{reso23}).
For Eq.(\ref{A11}), we simply substitute $\mu=1$, $\nu=0$  in 
Eq.(\ref{A1}).
\end{proof}

\subsection{Analogy between the function $F(x_1,x_2)$  and the Wigner function of the harmonic oscillator}\label{WWW}

There is similarity between the function $F(x_1,x_2)$ which describes a system with finite-dimensional Hilbert space, and the Wigner function of the harmonic oscillator which involves continuous variables.
We note here that the Wigner function for  a system with finite-dimensional Hilbert space is a discrete structure (it involves discrete variables that take a finite number of values).
Our approach describes systems with finite-dimensional Hilbert space with continuous variables, that are related to the probabilities associated with random projectors.

Using the notation in the review of ref. \cite{V2}, we define in the harmonic oscillator context
the displaced parity operators
\begin{eqnarray}
&&{\cal U}(A)=D(A){\cal U}(0)[D(A)]^{\dagger};\;\;\;D(A)=\exp (Aa^{\dagger}-A^*a);\;\;\;A\in {\mathbb C}\nonumber\\
&&{\cal U}(0)=\sum _{N=0}^\infty(-1)^N\ket{N}\bra{N};\;\;\;{\cal U}(A)[{\cal U}(A)]^\dagger ={\bf 1},
\end{eqnarray}
where $a^{\dagger}, a$ are creation and annihilation operators, $\ket{N}$ are number eigenstates, and ${\cal U}(0)$ is the parity operator around the origin.
They obey the `marginal relations'
\begin{eqnarray}\label{XX2}
&&\int _{-\infty}^{\infty} dA_R {\cal U}(A)=\frac{\pi}{\sqrt 2}\ket{A_I\sqrt 2}_p\;_p\bra {A_I\sqrt 2};\;\;\;\int _{-\infty}^{\infty} dA_I {\cal U}(A)=\frac{\pi}{\sqrt 2}\ket{A_R\sqrt 2}_x\;_x\bra {A_R\sqrt 2}\nonumber\\
&&\int _{-\infty}^{\infty} \int _{-\infty}^{\infty}  \frac{d^2A}{\pi}{\cal U}(A)={\bf 1}.
\end{eqnarray}
Here $\ket{}_x$, $\ket {}_p$ are position and momentum states, and $A_R, A_I$ are the real and imaginary parts of $A$.
The $\ket{A_I\sqrt 2}_p\;_p\bra {A_I\sqrt 2}$, $\ket{A_R\sqrt 2}_x\;_x\bra {A_R\sqrt 2}$ are projectors, and the ${\cal U}(A)$ is a unitary operator.
The Wigner function is defined in terms of the displaced parity operator ${\cal U}(A)$ as
\begin{eqnarray}
W(A)={\rm Tr}[\rho {\cal U}(A)].
\end{eqnarray}

Comparison of Eqs(\ref{XX1}), (\ref{XX2}), shows that there is analogy between ${\cal T}(x_2){\cal T}(x_1)$ and the displaced parity operators  ${\cal U}(A)$, and consequently between the 
function $F(x_1, x_2)$ and the Wigner function. Of course there are differences, like the fact that ${\cal U}(A)$ are unitary operators, while the ${\cal T}(x_2){\cal T}(x_1)$ are not unitary operators.

The marginals of the Wigner function
\begin{eqnarray}\label{wigner}
&&\int _{-\infty}^{\infty} dA_R W(A)=\frac{\pi}{\sqrt 2}\;_p\bra {A_I\sqrt 2}\rho \ket{A_I\sqrt 2}_p;\;\;\;\int _{-\infty}^{\infty} dA_I W(A)=\frac{\pi}{\sqrt 2}\;_x\bra {A_R\sqrt 2}\rho \ket{A_R\sqrt 2}_x\nonumber\\
&&\int _{-\infty}^{\infty} \int _{-\infty}^{\infty}\frac{d^2A}{\pi}W(A)={1},
\end{eqnarray}
are analogous to the marginals ${\mathfrak F}(x_1)$, ${\mathfrak F}(x_2)$ of the function $F(x_1, x_2)$.

We note that
\begin{itemize}
\item
The harmonic oscillator Wigner function
shows the location of a quantum state in quantum phase space.
This is done using the  partition of the unity by the displaced parity operators ${\cal U}(A)$ in Eq.(\ref{XX2}).
Large values of the Wigner function $W(A)$ show which ${\cal U}(A)$ overlap most with the density matrix.
The integral of the Wigner function over the phase-space plane is $1$. 
The two marginals are probability distributions.
\item
The operators  ${\cal T}(x_2){\cal T}(x_1)$ provide another partition of the unity.
Large values of the $F(x_1,x_2)$ show which ${\cal T}(x_2){\cal T}(x_1)$ overlap most with the density matrix.
The integral of the $F(x_1,x_2)$ over the quadrant $[0, \infty)\times [0,\infty)$ is $1$. 
The two marginal functions are equal to each other, and they are non-negative functions with integral equal to $1$, but they cannot be interpreted as probability distributions.
\end{itemize}

\subsection{Example}\label{ex1}

We consider a quantum system with positions and momenta in ${\mathbb Z}(3)$ (the integers modulo $3$), and $3$-dimensional Hilbert space $H(3)$.
Let $|{X};\nu\rangle$ and $|{P};\nu \rangle$ be orthonormal bases which we call position and momentum states, correspondingly.
They are related through a Fourier transform:
\begin{eqnarray}\label{PPPT}
|{P};\nu\rangle=\frac{1}{\sqrt 3}\sum _{\mu, \nu}\omega ^{\mu \nu}\ket{X;\mu};\;\;\;
\omega=\exp \left (i\frac {2\pi }{3}\right );\;\;\;\mu, \nu \in{\mathbb Z}(3).
\end{eqnarray}
The $X, P$ in the notation are not variables but they simply indicate position and momentum states.

We also consider the following pre-basis of $n=4$ vectors:
\begin{eqnarray}
\ket{1}=\frac{1}{\sqrt 2}
\begin{pmatrix}
1\\
1\\
0\\
\end{pmatrix};\;\;\;
\ket{2}=\frac{1}{\sqrt 2}
\begin{pmatrix}
1\\
0\\
1\\
\end{pmatrix};\;\;\;
\ket{3}=\frac{1}{\sqrt 2}
\begin{pmatrix}
0\\
1\\
1\\
\end{pmatrix};\;\;\;
\ket{4}=\frac{1}{\sqrt 6}
\begin{pmatrix}
1\\
1\\
2\\
\end{pmatrix},
\end{eqnarray}
They are written in terms of components in the basis of position states.
In this case
\begin{eqnarray}\label{PR1}
\Pi (\{1\})=\frac{1}{2}
\begin{pmatrix}
1&1&0\\
1&1&0\\
0&0&0\\
\end{pmatrix};\;\;
\Pi (\{2\})=\frac{1}{2}
\begin{pmatrix}
1&0&1\\
0&0&0\\
1&0&1\\
\end{pmatrix};\;\;
\Pi (\{3\})=\frac{1}{2}
\begin{pmatrix}
0&0&0\\
0&1&1\\
0&1&1\\
\end{pmatrix};\;\;
\Pi (\{4\})=\frac{1}{6}
\begin{pmatrix}
1&1&2\\
1&1&2\\
2&2&4\\
\end{pmatrix}.
\end{eqnarray}
Also using Eq.(\ref{670A}), we find
\begin{eqnarray}\label{PR2}
&&\Pi (\{1,2\})=\frac{1}{3}
\begin{pmatrix}
2&1&1\\
1&2&-1\\
1&-1&2\\
\end{pmatrix};\;\;
\Pi (\{1,3\})=\frac{1}{3}
\begin{pmatrix}
2&1&-1\\
1&2&1\\
-1&1&2\\
\end{pmatrix};\;\;
\Pi (\{1,4\})=\frac{1}{2}
\begin{pmatrix}
1&1&0\\
1&1&0\\
0&0&2\\
\end{pmatrix}\nonumber\\
&&\Pi (\{2,3\})=\frac{1}{3}
\begin{pmatrix}
2&-1&1\\
-1&2&1\\
1&1&2\\
\end{pmatrix};\;\;
\Pi (\{2,4\})=\frac{1}{3}
\begin{pmatrix}
2&-1&1\\
-1&2&1\\
1&1&2
\end{pmatrix};\;\;
\Pi (\{3,4\})=\frac{1}{3}
\begin{pmatrix}
2&-1&1\\
-1&2&1\\
1&1&2
\end{pmatrix}.
\end{eqnarray}
Furthermore
\begin{eqnarray}\label{PR3}
\Pi(\emptyset)=0;\;\;\;
\Pi (\{1,2,3\})=\Pi (\{1,2,4\})=\Pi (\{1,3,4\})=\Pi (\{2,3,4\})=\Pi (\{1,2,3,4\})={\bf 1}.
\end{eqnarray}
Using Eq.(\ref{ppp}) we get 
$\Lambda_{1}(i)=\Pi(\{i\})$ and
\begin{eqnarray}\label{mmm}
&&\Lambda _{2}(1)=\frac{1}{6}
\begin{pmatrix}
7&6&-5\\
6&7&-5\\
-5&-5&4\\
\end{pmatrix};\;\;
\Lambda_{3}(1)=
\begin{pmatrix}
1&1&-1\\
1&1&-1\\
-1&-1&1\\
\end{pmatrix}\nonumber\\&&
\Lambda_{2}(2)=\frac{1}{6}
\begin{pmatrix}
8&-6&4\\
-6&5&-3\\
4&-3&5\\
\end{pmatrix};\;\;
\Lambda_{3}(2)=\frac{1}{6}
\begin{pmatrix}
7&-3&0\\
-3&7&-4\\
0&-4&4\\
\end{pmatrix}\nonumber\\
&&\Lambda_{2}(3)=\frac{1}{6}
\begin{pmatrix}
5&-6&-3\\
-6&8&4\\
-3&4&5\\
\end{pmatrix};\;\;
\Lambda_{3}(3)=\frac{1}{6}
\begin{pmatrix}
7&-3&-4\\
-3&7&0\\
-4&0&4\\
\end{pmatrix}\nonumber\\&&
\Lambda_{2}(4)=\frac{1}{6}
\begin{pmatrix}
5&-4&1\\
-4&5&1\\
1&1&8\\
\end{pmatrix};\;\;
\Lambda_{3}(4)=\frac{1}{3}
\begin{pmatrix}
3&-1&-1\\
-1&3&-1\\
-1&-1&3\\
\end{pmatrix}.
\end{eqnarray}
They give  the density matrices
\begin{eqnarray}
\sigma _{\Omega}(i)=\frac{4}{3}\sum _{a=1}^3 B(a,4-a+1)\Lambda _{a}(i);\;\;\;\frac{3}{4}\sum _{i=1}^4\sigma _{\Omega}(i)={\bf 1},
\end{eqnarray}
where
\begin{eqnarray}\label{123}
&&\sigma _\Omega (1)=
\begin{pmatrix}
0.407&0.388&-0.203\\
0.388&0.407&-0.203\\
-0.203&-0.203&0.185\\
\end{pmatrix};\;\;
\sigma _\Omega (2)=
\begin{pmatrix}
0.444&-0.166&0.240\\
-0.166&0.222&-0.129\\
0.240&-0.129&0.333\\
\end{pmatrix}\nonumber\\&&
\sigma _\Omega (3)=
\begin{pmatrix}
0.222&-0.166&-0.129\\
-0.166&0.444&0.240\\
-0.129&0.240&0.333\\
\end{pmatrix};\;\;
\sigma _\Omega (4)=
\begin{pmatrix}
0.259&-0.055&0.092\\
-0.055&0.259&0.092\\
0.092&0.092&0.481\\
\end{pmatrix}.
\end{eqnarray}

The density matrices ${R _a}$ are calculated from the matrices $\Lambda _{a}(i)$ given in Eq.(\ref{mmm}), in the position basis:
\begin{eqnarray}
R_1=
\begin{pmatrix}
0.291&0.166&0.208\\
0.166&0.291&0.208\\
0.208&0.208&0.416
\end{pmatrix};\;
R_2=
\begin{pmatrix}
0.347&-0.138&-0.041\\
-0.138&0.347&-0.041\\
-0.041&-0.041&0.305
\end{pmatrix};\;
R_3=
\begin{pmatrix}
0.361&-0.027&-0.166\\
-0.027&0.361&-0.166\\
-0.166&-0.166&0.277
\end{pmatrix}
\end{eqnarray}
They resolve the identity:
\begin{eqnarray}
R_1+R_2+R_3={\bf 1}.
\end{eqnarray}
Then the operators ${\cal T}(x)$ are:
 \begin{eqnarray}
{\cal T}(x)&=&\frac{4{R _1}+12x{R _2}+12x^2{R _3}}{(1+x)^5}
\end{eqnarray}

Using Eq.(\ref{cdb}), we expand the position states as
\begin{eqnarray}\label{1234}
&&\ket{X;\nu}=\int _0^\infty \frac{dx}{(1+x)^5}
\begin{pmatrix}
\alpha _{\nu 0}+\beta_{\nu 0}x+\gamma _{\nu 0}x^2\\
\alpha _{\nu 1}+\beta_{\nu 1}x+\gamma _{\nu 1}x^2\\
\alpha _{\nu 2}+\beta_{\nu 2}x+\gamma _{\nu 2}x^2\\
\end{pmatrix}
\end{eqnarray}
where the coefficients are given in table \ref{t2}.
The momentum 
 states are expanded as
\begin{eqnarray}\label{1239}
&&\ket{P;\nu}=\int _0^\infty \frac{dx}{(1+x)^5}
\begin{pmatrix}
\widetilde \alpha _{\nu 0}+\widetilde\beta_{\nu 0}x+\widetilde\gamma _{\nu 0}x^2\\
\widetilde\alpha _{\nu 1}+\widetilde\beta_{\nu 1}x+\widetilde\gamma _{\nu 1}x^2\\
\widetilde\alpha _{\nu 2}+\widetilde\beta_{\nu 2}x+\widetilde\gamma _{\nu 2}x^2\\
\end{pmatrix}
\end{eqnarray}
where
\begin{eqnarray}
\widetilde \alpha _{\nu j}=\frac{1}{\sqrt 3}\sum _{\mu}\omega ^{\mu \nu}\alpha _{\mu j};\;\;\;\widetilde \beta _{\nu j}=\frac{1}{\sqrt 3}\sum _{\mu}\omega ^{\mu \nu}\beta _{\mu j};\;\;\;\widetilde \gamma _{\nu j}=\frac{1}{\sqrt 3}\sum _{\mu}\omega ^{\mu \nu}\gamma _{\mu j}.
\end{eqnarray}
The density matrix $\rho=\frac{1}{3}{\bf 1}$ is represented with the
\begin{eqnarray}
{\cal O}(x_1,x_2)=\frac{1}{3[(1+x_1)(1+x_2)]^{5}}\sum _{a=1}^3\sum _{b=1}^3 x_1^{a-1}x_2^{b-1}\frac{R_aR_b}{B(a,5-a)B(b,5-b)} .
\end{eqnarray}

We next calculate the function $F(x_1,x_2)$ for the position states. We find
\begin{eqnarray}\label{ddd}
\ket{X;\nu}\;\rightarrow\;F_\nu (x_1,x_2)=\frac{1}{(1+x_1)^5(1+x_2)^5}\sum _{i=0}^2\sum _{j=0}^2 A_{ij-\nu}x_1^ix_2^j,
\end{eqnarray}
where the coefficients are given in table \ref{t3}.

\begin{table}
\caption{The coefficients in Eqs.(\ref{ddd}).}
\def\arraystretch{2}
\begin{tabular}{|c||c|c|c|c|c|c|c|c|c|}\hline
$\nu$&$A_{00-\nu }$&$A_{10-\nu }$&$A_{01-\nu }$&$A_{20-\nu }$&$A _{02-\nu }$&$A _{11-\nu }$&$A _{21-\nu }$&$A _{12-\nu }$&$A _{22-\nu }$\\\hline
$0$&$2.497$&$3.331$&$3.331$&$3.164$&$3.164$&$20.381$&$19.606$&$19.606$&$22.885$\\\hline
$1$&$2.497$&$3.331$&$3.331$&$3.164$&$3.164$&$20.381$&$19.606$&$19.606$&$22.885$\\\hline
$2$&$4.163$&$5.274$&$5.274$&$2.220$&$2.220$&$13.939$&$14.218$&$14.218$&$19.108$\\\hline
\end{tabular} \label{t3}
\end{table}

\section{Discussion}

Resolutions of the identity are an important part of the quantum formalism (and also of the Harmonic Analysis formalism).
They underpin the theory of coherent states, POVM, and frames and wavelets.
The present work is a contribution to this general area.  We construct an infinite number of continuous resolutions of the identity, in a finite-dimensional Hilbert space.
Consequently, systems with finite-dimensional Hilbert space which are naturally described with discrete variables and `Discrete Mathematics', are described here with continuous variables and `Continuous Mathematics'.

 The first part of the paper  (sections 2, 4, 5) is in the context of set theory.
Random sets are introduced in definition \ref{def100}. Their average cardinalities are functions of $n$ probabilities (proposition \ref{PRO2}), and they  are used to get a continuous partition of the cardinality of the union of $n$ overlapping sets (proposition \ref{PRO10}).
The formalism uses M\"obius transforms and it is inspired by Shapley's methodology in cooperative game theory.
An example has been presented in section \ref{ex}.
This part of the paper exposes the methodology in a simpler setting than a Hilbert space.
It is interesting in its own right, but here we are interested in its generalisation 
in the second part of the paper (sections 3, 6, 7, 8) into the context of  finite-dimensional Hilbert spaces.

Random projectors  into the $2^n$ subspaces spanned by states from a total set of $n$ states are introduced in definition \ref{def200}.
Their average is an operator which depends on $n$ probabilities (proposition \ref{PRO11}) and is used to construct an infinite number of continuous resolutions of the identity (proposition \ref{PRO100}),
that involve Hermitian positive semi-definite operators.
The simplest one is the diagonal continuous resolution of the identity(proposition \ref{PRO100}), and it is used to expand an arbitrary vector 
in terms of a continuum of components (Eq.(\ref{cdb})). It is also used to define the $F(x_1,x_2)$ function 
which is analogous to the Wigner function for the harmonic oscillator. An example has been presented in section \ref{ex1}.

The work is a contribution to the area of resolutions of the identity and their applications.

\end{document}